\documentclass[journal, twoside]{IEEEtran}
\usepackage[T1]{fontenc}
\usepackage{balance, enumerate}
\usepackage{cite, xcolor}

\ifCLASSINFOpdf
   \usepackage[pdftex]{graphicx}
\else
\fi

\usepackage{amsmath, amssymb, mathtools, amsthm, cuted, xcolor, mathrsfs, lipsum, enumitem}
\usepackage[hyphens]{url}

\newtheorem{theorem}{Theorem}
\begin{document}
\title{Delay Violation Probability and Effective Rate of Downlink NOMA over $\alpha$-$\mu$ Fading Channels}

\author{Vaibhav~Kumar,~\IEEEmembership{Student Member,~IEEE,} Barry~Cardiff,~\IEEEmembership{Senior~Member,~IEEE,} \\ Shankar~Prakriya,~\IEEEmembership{Senior~Member,~IEEE,} and~Mark~F.~Flanagan,~\IEEEmembership{Senior Member,~IEEE}
\thanks{Copyright (c) 2015 IEEE. Personal use of this material is permitted. However, permission to use this material for any other purposes must be obtained from the IEEE by sending a request to pubs-permissions@ieee.org.}
\thanks{This publication has emanated from research conducted with the financial support of Science Foundation Ireland (SFI) and is co-funded under the European Regional Development Fund under Grant Number 13/RC/2077.}
\thanks{V. Kumar, B. Cardiff and M. F. Flanagan are with School of Electrical and Electronic Engineering, University College Dublin, Ireland (e-mail: vaibhav.kumar@ieee.org, barry.cardiff@ucd.ie, mark.flanagan@ieee.org).}
\thanks{S. Prakriya is with Department of Electrical Engineering, Indian Institute of Technology -- Delhi, India (email: shankar@ee.iitd.ac.in).}}

\markboth{IEEE Transactions on Vehicular Technology}%
{Kumar \MakeLowercase{\textit{et al.}}: On the Effective Capacity of Downlink NOMA over $\alpha$-$\mu$ Fading Channels}

\maketitle
\let\bs\boldsymbol

\begin{abstract}
Non-orthogonal multiple access (NOMA) is a potential candidate to further enhance the spectrum utilization efficiency in beyond fifth-generation (B5G) standards. However, there has been little attention on the quantification of the delay-limited performance of downlink NOMA systems. In this paper, we analyze the performance of a two-user downlink NOMA system over generalized $\alpha$-$\mu$ fading in terms of delay violation probability and effective rate (ER). In particular, we derive an analytical expression for an upper bound on the delay violation probability and we derive the exact sum ER of the downlink NOMA system. We also derive analytical expressions for high and low-SNR approximations to the sum ER, as well as a fundamental upper bound on the sum ER which represents the ergodic sum-rate for the downlink NOMA system. We also analyze the sum ER of a corresponding time-division-multiplexed orthogonal multiple access (OMA) system. Our results show that while NOMA consistently outperforms OMA over the practical SNR range, the relative gain becomes smaller in more severe fading conditions, and is also smaller in the presence a more strict delay QoS constraint.  
\end{abstract}

\IEEEpeerreviewmaketitle

\section{Introduction}
\IEEEPARstart{N}{OMA} has drawn tremendous attention as a potential solution to the spectrum congestion problem in 5G and B5G wireless communication systems. However, the majority of available literature related to NOMA primarily considers the system performance based on Shannon-perspective metrics such as achievable rate and outage probability, as well as performance optimization using these metrics. In contrast to this, a delay-constrained QoS metric termed as \emph{effective rate}\footnote{This is sometimes also referred to as ``effective capacity'' or ``link-layer capacity''.} (ER) was introduced in~\cite{Negi}, which is defined as the maximum constant arrival rate that a service process can support under a statistical delay constraint. The \emph{delay violation probability} (DVP) is a related metric, which indicates the probability that an information bit arriving at the source will not be successfully delivered to the destination within a certain delay bound. The ER performance of orthogonal multiple access (OMA) based wireless communication systems over different fading channel models has been well-studied in the past, but in the case of NOMA, the literature is mostly limited to analysis over Rayleigh, Rician and Nakagami-$m$ fading channels.

The performance analysis of OMA systems, including single and multiple-antenna architectures over different fading channels in terms of ER, has been presented in~\cite{CR-Nakagami, Unified-MISO, Weibull, EtaMu, AlphaMu, KappaMuShadowed, Composite, Snedecor, VTC, LpNorm}. An optimal rate and power adaptation policy for the maximization of the ER of a secondary/unlicensed user in an average-interference-constrained underlay spectrum sharing system over Nakagami-$m$ fading channels was presented in~\cite{CR-Nakagami}. Using a moment generating function approach, the ER analysis of a single user with independent but not necessarily identically distributed multiple-input single-output (MISO) channel assuming hyper Fox's $H$ as well as generalized-$\mathcal K$ fading was presented in~\cite{Unified-MISO}. Some other notable contributions also exist on the analysis of the ER for a single-user communication system over more flexible/generalized channel models including (MISO) Weibull~\cite{Weibull}, (MISO) $\eta$-$\mu$~\cite{EtaMu}, (MISO) $\alpha$-$\mu$~\cite{AlphaMu}, $\kappa$-$\mu$ shadowed~\cite{KappaMuShadowed}, composite $\alpha$-$\eta$-$\mu$/Gamma~\cite{Composite} and (MISO) Fisher-Snedecor $\mathcal F$ fading~\cite{Snedecor}. Moreover, the ER analysis of a single-user equal gain combining receiver over $\alpha$-$\kappa$-$\mu$ and $\alpha$-$\eta$-$\mu$ fading channels was presented in~\cite{VTC}. Very recently, the ER analysis of a single-user $L_p$-norm diversity combiner with adaptive transmission schemes over arbitrary generalized fading channels was presented in~\cite{LpNorm}.

The impact of finite blocklength coding and \emph{age of information} (AoI) on DVP, as well as the DVP analysis for a multiuser MISO OMA system, were presented in~\cite{ACM_Paper, FiniteBlockLength-DelayViolation, MISO-FiniteBlockLength-DelayViolation, AoI}. Using stochastic network calculus (SNC) and the Mellin transform, the impact of finite blocklength coding on the DVP of a point-to-point single-user communication system over the Rayleigh fading channel was analyzed in~\cite{ACM_Paper}. The optimal design criteria to minimize the DVP in a point-to-point communication system over the Rayleigh fading channel using finite blocklength coding were presented in~\cite{FiniteBlockLength-DelayViolation}. The impact of finite blocklength coding on the DVP in a multiuser MISO downlink OMA system over Rayleigh fading channels with imperfect CSI was presented in~\cite{MISO-FiniteBlockLength-DelayViolation}. The impact of average AoI on the DVP in a two-user time-slotted multiple access Rayleigh fading channel was presented in~\cite{AoI}. In~\cite{LargeWirelessNetworks}, the DVP analysis for a large wireless network over Rayleigh fading channels was presented, where the transmitter locations are modeled using a Poisson point process.

The delay-constrained performance analysis of NOMA systems over conventional (Rayleigh, Rician and Nakagami-$m$) fading channels was presented in~\cite{SubOptimal, QoS_Provisioning_Ding, Half-FullDuplex-ER, Mussavian, IoT-NOMA, NOMA_JSTSP, Poor, NOMA_Secrecy, Uplink-Musavian, FiniteBlockLength-Musavian, UplinkNOMA-FiniteBlockLength}. The analysis of the sum ER for a two-user downlink NOMA system in Rayleigh fading was first presented in~\cite{SubOptimal}, where a suboptimal power allocation policy was proposed for the maximization of the sum ER using a truncated channel inversion power control policy. A cross-layer optimal power allocation scheme for a half-duplex cooperative multi-user downlink NOMA system over Rayleigh fading channels was suggested in~\cite{QoS_Provisioning_Ding}, where the optimization problem was shown to be quasi-concave and a bisection-based solution was proposed in order to maximize the ER of the user having the minimum ER. Building on~\cite{QoS_Provisioning_Ding}, a similar \emph{max-min} criterion based optimal power allocation scheme was suggested for a half-duplex/full-duplex cooperative two-user downlink NOMA system over Rayleigh fading channels in~\cite{Half-FullDuplex-ER}. A comprehensive analysis of the ER for a multiuser downlink NOMA system with a per-user delay QoS requirement over Rayleigh fading was presented in~\cite{Mussavian} and it was shown that the OMA system achieves a higher sum ER in the low signal-to-noise ratio (SNR) regime, whereas the NOMA system prevails in the high-SNR regime. The performance analysis of a two-user uplink NOMA system over Rayleigh fading channels in terms of DVP was presented in~\cite{IoT-NOMA}, where the authors also presented an optimal transmit power allocation scheme. The ER and DVP analysis for  multiuser\footnote{In~\cite{NOMA_JSTSP}, although the authors considered a multiuser scenario, it was assumed that the users are grouped into multiple NOMA \emph{pairs}. The analysis and the results were presented for an arbitrary NOMA pair while the problem of optimal user-pairing was not considered.} downlink NOMA over Nakagami-$m$ and Rician fading using SNC and the Mellin transform was presented in~\cite{NOMA_JSTSP}. In that work, two different optimal power allocation schemes, namely MaxMinEC and MinMaxDVP, were also suggested to maximize the minimum ER and minimize the maximum DVP, respectively. In order to analyze the performance of a two-user downlink NOMA system under statistical QoS constraints in the finite-blocklength regime, and to reduce the complications associated with the ER maximization problem due to the existence of statistical delay-bounded and error-rate-bounded QoS constraints, the concept of ``$\epsilon$-effective capacity'' was introduced in~\cite{Poor}. For the analysis of the delay-constrained secrecy performance of a multiuser downlink NOMA system over Rayleigh fading channels, the notion of \emph{effective secrecy rate} was introduced in~\cite{NOMA_Secrecy}. A closed-form expression for the sum ER of a two-user uplink NOMA system over Rayleigh fading was presented in~\cite{Uplink-Musavian}. A closed-form expression for the sum ER in a multiuser\footnote{In~\cite{FiniteBlockLength-Musavian}, the authors considered a multiuser system consisting of multiple NOMA pairs, and showed that the NOMA pairs with more distinct channel conditions attain a higher sum ER.} downlink NOMA system in the finite-blocklength regime over Rayleigh fading channels was presented in~\cite{FiniteBlockLength-Musavian}. The analysis of DVP for a two-user uplink NOMA system using finite blocklength coding over Rayleigh fading channels, considering imperfect channel state information (CSI). was presented in~\cite{UplinkNOMA-FiniteBlockLength}.

To the best of our knowledge, the analysis of DVP and ER for NOMA systems in generalized fading channels have not yet been presented in the literature. In this paper, we analyze the performance of a two-user downlink NOMA system in terms of the DVP and the sum ER over $\alpha$-$\mu$ fading. The $\alpha$-$\mu$ distribution is a mathematically tractable and generalized distribution which provides a remarkable unification of the system's performance over a variety of fast fading channel models including one-sided Gaussian, Chi-squared, Rayleigh, Nakagami-$m$, Erlang and Weibull fading. 
Against this background, the main contributions of this paper are listed below:
\begin{itemize}
	\item We analyze the DVP of a two-user downlink NOMA system over $\alpha$-$\mu$ fading channels. Using SNC and the Mellin transform, we derive a closed-form upper bound on the DVP for each user, and use these expressions to demonstrate the effect of different channel parameters on the DVP.
	\item We also derive exact closed-form expressions for the ER of each user in the downlink NOMA system\footnote{Note that this is different from~\cite{NOMA_JSTSP}, where approximate expressions were only derived for the ER of the \emph{weak} user over Nakagami-$m$ and Rician fading.} as well as a corresponding downlink OMA system over $\alpha$-$\mu$ channels, and show the effect of non-linearity, clustering, and delay parameters on the difference between the sum ER in NOMA and OMA systems. Furthermore, we also derive approximate expressions for the sum ER of downlink NOMA in the high-SNR and low-SNR\footnote{Note that for low-SNR approximation in~\cite{NOMA_JSTSP}, a first-order Taylor approximation of the form $\ln(1 + x)\approx x; x \to 0$ was used. In this paper, we use a second-order Taylor approximation, i.e., $f(x) \approx x f'(0) + 0.5 x^2 f''(0); \ x \to 0$, which helps in quantifying the minimum energy per bit required for reliable transmission of information and the wideband slope.} regimes.
	\item Moreover, using Jensen's inequality, we derive an upper-bound on the sum ER of downlink NOMA which is independent of the delay constraint, which represents the ergodic sum-rate of the NOMA system. By evaluating the difference between the ergodic sum-rate and sum ER (which we term as the \emph{rate loss}) in the downlink NOMA system, we demonstrate that the rate loss increases with an increase in the SNR, the delay exponent and the severity of fading.
\end{itemize}
\section{System Model}	
We consider a half-duplex downlink NOMA scenario, where a source $S$ simultaneously communicates with two users $U_s$ (the user with \emph{strong} average channel condition) and $U_w$ (the user with \emph{weak} average channel condition). The channel fading coefficient of the $S$-$U_i$ link ($i \in \{s, w\}$) is denoted by $h_i$, and the corresponding channel gain is denoted by $g_i$ ($\triangleq |h_i|^2$) with $\alpha$-root-mean value denoted by $\Omega_i$. It is assumed that both wireless links are $\alpha$-$\mu$ distributed with \emph{non-linearity} parameter denoted by $\alpha$ and \emph{clustering} parameter denoted by $\mu$~\cite{Yacoub}. Assuming $\alpha, \mu \in \mathbb Z_+$ (the set of all positive integers), and using~\cite[eqns.~(1),~(2)]{ICC} and~\cite[eqn.~(8.352-1),~p.~899]{Grad}, the cumulative distribution function (CDF) and probability density function (PDF) of the channel gain $g_i$ are, respectively, given by 
\begin{align}
F_{g_i}(x) = 1 - \exp \left( - \frac{\mu x^{0.5 \alpha}}{\Omega_i^{\alpha}}\right) \sum_{j = 0}^{\mu - 1} \frac{ \mu^j x^{0.5 j \alpha}}{j! \Omega_i^{j \alpha}}, \label{F_gi}
\end{align}
\begin{align}
f_{g_i}(x) = \dfrac{ \alpha \mu^{\mu} x^{0.5 \alpha \mu - 1} }{ 2 \Omega_i^{\alpha \mu} \Gamma(\mu) } \exp \left( -\frac{ \mu x^{0.5 \alpha} }{\Omega_i^{\alpha} }\right), \label{f_gi}
\end{align}
where $\Gamma(\cdot)$ denotes the Gamma function. Furthermore, we assume that $\Omega_w^{\alpha} < \Omega_s^{\alpha}$. We also assume that (perfect) \emph{instantaneous} CSI is available at $U_i$ regarding the $S$-$U_i$ link, while only the \emph{statistical} CSI (comprising the values of $\Omega_s$ and $\Omega_w$) is available at $S$.

In the case of downlink NOMA, the source $S$ transmits a superimposed symbol $\sum_i \sqrt{a_i \mathcal E} x_i$ to both users, where $\mathcal E$ denotes the energy budget per (superposed) symbol at the source, $a_i$ denotes the power-allocation coefficient for user $U_i$ with $a_s < a_w$ and $\sum_i a_i = 1$ , and $x_i$ represents the data-bearing (complex) constellation symbol intended for user $U_i$ with unit energy. The received signal at user $U_i$ is given by $y_i = h_i (\sqrt{a_s \mathcal E} x_s + \sqrt{a_w \mathcal E} x_w) + z_i$, where $z_i \sim \mathcal{CN}(0, \sigma^2)$ denotes complex additive white Gaussian noise (AWGN). Upon reception of the signal, $U_w$ decodes $x_w$ treating the interference due to the presence of $x_s$ as additional noise. On the other hand, $U_s$ also decodes $x_w$ first (by treating the interference due to $x_s$ as additional noise) and then applies successive interference cancellation (SIC) to remove $x_w$ from $y_s$, and subsequently decodes $x_s$. Therefore, the instantaneous received signal-to-interference-plus-noise ratio (SINR) and SNR to decode $x_w$ and $x_s$ are, respectively, given by $\gamma_w = a_w \rho g_{\min}/(a_s \rho g_{\min} + 1)$ and $\gamma_s = a_s \rho g_s$, where $g_{\min} \triangleq \min\{g_w, g_s\}$ and $\rho \triangleq \mathcal E/\sigma^2$ (we will refer to $\rho$ as the SNR in the remainder of this paper). 

Using~\eqref{F_gi},~\eqref{f_gi} and a transformation of random variables, the PDF of $g_{\min}$ can be given by 
\begin{multline} \label{f_gmin}
\!\!\!\! f_{g_{\min}}(x) = \frac{\alpha \exp \left( -\mu x^{0.5 \alpha}/\tilde \Omega \right)}{2 \Omega_s^{\alpha \mu} \Gamma(\mu)}  \sum_{m = 0}^{\mu - 1} \frac{\mu^{(\mu + m)} x^{0.5 \alpha(\mu + m) - 1}}{m! \Omega_w^{m \alpha}} \\
 + \dfrac{\alpha \exp \left( -\mu x^{0.5 \alpha}/\tilde \Omega \right)}{2 \Omega_w^{\alpha \mu} \Gamma(\mu)} \sum_{n = 0}^{\mu - 1} \dfrac{\mu^{(\mu + n)} x^{0.5 \alpha(\mu + n) - 1}}{n! \Omega_s^{n \alpha}}, 
\end{multline}
where $\tilde \Omega = 1/\left( \tfrac{1}{\Omega_s^{\alpha}} + \tfrac{1}{\Omega_w^{\alpha}}\right)$.
\section{Performance analysis}
\subsection{Delay violation probability in downlink NOMA}
Assuming that the instantaneous traffic arrival for user $U_i \in \{s, w\}$ at $S$ in the $k$-th time slot is denoted by $\lambda_i(k)$ and the corresponding traffic departure is denoted by $d_i(k)$, the cumulative arrival and departure processes from time slot $\tau$ to $t-1$ can be represented by $A_i(\tau, t) = \sum_{k = \tau}^{t-1} \lambda_i(k)$ and $D_i(\tau, t) = \sum_{k = \tau}^{t-1}d_i(k)$, respectively. Therefore, the delay experienced in the successful delivery of the information bits (intended for user $U_i$) arriving at $S$ in time slot $t$ is given by 
\begin{align*}
\vartheta_i(t) \triangleq \inf \{u \geq 0: A_i(0, t) \leq D_i(0, t + u)\}.
\end{align*}
Using~\cite[eqn.~(17)]{NOMA_JSTSP}, an upper bound on the probability of violation of a target delay $\vartheta$ is given by 
\begin{align}
\Pr(\vartheta_i > \vartheta) \leq \inf_{\mathscr S > 0} \left\{\dfrac{[\mathcal M_{\varphi_i} (1 - \mathscr S)]^{\vartheta}}{1 - \exp(\lambda \mathscr S) \mathcal M_{\varphi_i} (1 - \mathscr S)}\right\}, \label{DelayProbBound}
\end{align}
where $\mathcal M_X(\mathscr S, \tau, t) = \mathbb E\{[X(\tau, t)]^{\mathscr S - 1}\}$ represents the Mellin transform of the random process $X(\tau, t)$, $\varphi_s \triangleq (1 + a_s \rho g_s)^{N/\ln 2}$, $\varphi_w \triangleq [1 + a_w \rho g_{\min}/ (1 + a_s \rho g_{\min})]^{N/\ln 2}$ and $N$ denotes the number of symbols transmitted by $S$ in each time slot. 
\begin{theorem} \label{Theorem_Mellin_s}
	For the case of $U_s$, an analytical expression for $\mathcal M_{\varphi_s} (1 - \mathscr S)$ can be given by
{\color{black}
	\begin{multline} \label{Mellin_s}
		\mathcal M_{\varphi_s}(1 - \mathscr S) \\ = \dfrac{\alpha^{\varpi} \mu^{\mu} G_{\alpha, 2 + \alpha}^{2 + \alpha, \alpha} \left[ \left. \tfrac{\mu^2}{4 (a_s \rho)^{\alpha} \Omega_s^{2\alpha}} \right\vert \begin{smallmatrix} \Delta(\alpha, 1 - 0.5 \alpha \mu) \\[0.1em] \Delta(2, 0), \ \Delta (\alpha, \varpi - 0.5 \alpha \mu) \end{smallmatrix} \right]}{ \sqrt{2} (2\pi)^{\alpha - 0.5} \Omega_s^{\alpha \mu} \Gamma(\mu) \Gamma(\varpi) (a_s \rho)^{0.5 \alpha \mu}},
	\end{multline}}\\
where $\varpi \triangleq N \mathscr S/\ln 2$, $\Delta(x, y) = \tfrac{y}{x}, \tfrac{y + 1}{x}, \ldots, \tfrac{y + x - 1}{x}$, and $G[\cdot]$ denotes Meijer's G-function.
\end{theorem}
\begin{proof}
	See Appendix~\ref{Appendix_Theorem_Mellin_s}.
\end{proof} 
Substituting the expression for $\mathcal M_{\varphi_s}(1 - \mathscr S)$ into~\eqref{DelayProbBound}, one can obtain an analytical expression for an upper bound on the DVP for $U_s$.
\begin{theorem} \label{Theorem_Mellin_w}
For the case of the weak user, an analytical expression for $\mathcal M_{\varphi_w}(1 - \mathscr S)$ can be given by 
\begin{multline} \label{Mellin_w}
	\mathcal M_{\varphi_w}(1 - \mathscr S) = \dfrac{1} {  \Gamma(\mu) \Gamma(\varpi) \Gamma(-\varpi) } \left[ \dfrac{1}{\Omega_s^{\alpha \mu}} \sum_{m = 0}^{\mu - 1} \dfrac{{\tilde \Omega}^{\mu + m}}{m! \Omega_w^{m \alpha}} \right.  \\
	\left. \times \mathcal H (\varpi, m) + \dfrac{1} { \Omega_w^{\alpha \mu} } \sum_{n = 0}^{\mu - 1} \dfrac{{\tilde \Omega}^{\mu + n}}{n! \Omega_s^{n \alpha}} \mathcal H (\varpi, n) \right],
\end{multline}
where $\mathcal H(x, y)$ denotes the extended generalized bivariate Fox's H-function (EGBFHF)\footnote{Computationally efficient implementations of EGBFHF using \textsc{Mathematica} and \textsc{Matlab} are given in~\cite{Secrecy} and~\cite{MatlabImplement}, respectively.}, i.e., 
\begin{multline*}
\!\!\!\!\!\mathcal H(x, y) = H_{\bs a}^{\bs b} \left[ \!\! \left. \begin{smallmatrix} \left( 1-(\mu + y); \tfrac{2}{\alpha}, \tfrac{2}{\alpha} \right)\\[0.6em] -\end{smallmatrix}\right\vert \left. \begin{smallmatrix} (1-x, 1)\\[0.6em] (0, 1)\end{smallmatrix} \right\vert \left. \begin{smallmatrix} (1+x, 1)\\[0.6em] (0, 1)\end{smallmatrix} \right \vert \right. \\
\left. \rho \left( \dfrac{\tilde \Omega}{\mu} \right)^{2/\alpha}, a_s \rho \left( \dfrac{\tilde \Omega}{\mu} \right)^{2/\alpha} \right], 
\end{multline*}
$\bs a = 1, 0:1, 1:1, 1$ and $\bs b = 0, 1:1, 1:1, 1$.
\end{theorem}
\begin{proof}
	See Appendix~\ref{Appendix_Theorem_Mellin_w}.
\end{proof}
Substituting the expression for $\mathcal M_{\varphi_w}(1 - \mathscr S)$ into~\eqref{DelayProbBound}, one can obtain an analytical expression for an upper bound on the DVP for $U_w$. It is important to note that in~\cite{NOMA_JSTSP}, an approximate expression for $f_{g_{\min}}(x)$ was used to obtain the analytical expression for $\mathcal M_{\varphi_w}(1 - \mathscr S)$ (over Nakagami-$m$ and Rician fading); in contrast, we use the \emph{exact} expression for $f_{g_{\min}}(x)$ to find the closed-form expression for $\mathcal M_{\varphi_w}(1 - \mathscr S)$ over $\alpha$-$\mu$ fading.
\subsection{Effective rate of downlink NOMA}
The ER for $U_i$ can be defined as~(c.f.~\cite{Mussavian})
	\begin{align}
		R_{i, \mathrm{NOMA}} = -\dfrac{1}{\nu} \log_2 \left[ \mathbb E_{\gamma_i}\{ (1 + \gamma_i)^{-\nu}\} \right], \label{Ei_def}
	\end{align}
where $\nu \triangleq \theta T B / \ln 2$, $\theta$ represents the delay QoS exponent, $T$ denotes the length of each fading block (which is assumed to be same for both $S$-$U_s$ and $S$-$U_w$ links and is an integer multiple of the symbol duration, with a further assumption that the durations of both user symbols are the same), and $B$ denotes the total bandwidth. 

Therefore, using~\eqref{Ei_def}, the ER for $U_s$ is given by 
\begin{align*}
R_{s, \mathrm{NOMA}} = \dfrac{-1}{\nu} \log_2 \left[ \int_0^\infty (1 + a_s \rho x)^{-\nu} f_{g_s}(x) \mathrm dx \right].
\end{align*}
Substituting the expression for $f_{g_s}(x)$ from~\eqref{f_gi} into the preceding equation and following the arguments in~Appendix~\ref{Appendix_Theorem_Mellin_s}, an analytical expression for the ER of $U_s$ is given by
	\begin{multline} \label{Es_exact_Closed}
		R_{s, \mathrm{NOMA}} \\
		=  -\dfrac{1}{\nu} \log_2 \left\{ \dfrac{\alpha^{\nu} \mu^{\mu} G_{\alpha, 2 + \alpha}^{2 + \alpha, \alpha} \left[ \left. \tfrac{\mu^2}{4 (a_s \rho)^{\alpha} \Omega_s^{2\alpha}} \right\vert \begin{smallmatrix} \chi \\[0.1em] \Delta(2, 0), \ \xi \end{smallmatrix} \right]}{ \sqrt{2} (2\pi)^{\alpha - 0.5} \Omega_s^{\alpha \mu} \Gamma(\mu) \Gamma(\nu) (a_s \rho)^{0.5 \alpha \mu}}  \right\}, 
	\end{multline}
where $\chi = \Delta(\alpha, 1 - 0.5 \alpha \mu)$ and $\xi = \Delta(\alpha, \nu - 0.5 \alpha \mu)$. On the other hand, using~\eqref{Ei_def}, the ER of $x_w$ can be given by 
\begin{multline*}
R_{w, \mathrm{NOMA}} \\
= \dfrac{-1}{\nu} \log_2 \left[\int_0^\infty (1 + \rho x)^{-\nu} (1 + a_s \rho x)^{\nu} f_{g_{\min}}(x) \mathrm dx \right].
\end{multline*}
Substituting the expression for $f_{g_{\min}}(x)$ from~\eqref{f_gmin} in the preceding expression and following the arguments in~Appendix~\ref{Appendix_Theorem_Mellin_w}, an analytical expression for the ER of $U_w$ is given by 
	\begin{multline}
		R_{w, \mathrm{NOMA}} = \dfrac{-1}{\nu} \log_2 \left[ \dfrac{1} { \Gamma(\mu) \Gamma(\nu) \Gamma(-\nu) } \left( \dfrac{1}{\Omega_s^{\alpha \mu} } \right. \right. \\
		\!\!\!\! \times \left. \left. \sum_{m = 0}^{\mu - 1} \dfrac{{\tilde \Omega}^{\mu + m}}{m! \Omega_w^{m \alpha}} \mathcal H (\nu, m)  + \dfrac{1} { \Omega_w^{\alpha \mu} } \sum_{n = 0}^{\mu - 1} \dfrac{{\tilde \Omega}^{\mu + n}}{n! \Omega_s^{n \alpha}} \mathcal H (\nu, n)\right) \right]. \label{Ew_exact_Closed}
	\end{multline}
Therefore, the sum ER of the downlink NOMA system is given by adding~\eqref{Es_exact_Closed} and~\eqref{Ew_exact_Closed}, i.e., $R_{\mathrm{sum, NOMA}} = R_{s, \mathrm{NOMA}} + R_{w, \mathrm{NOMA}}$. Similar to the case of DVP, an approximate expression for $f_{g_{\min}}(x)$ was used to find the analytical expressions for the ER of the weak user over Nakagami-$m$ and Rician fading in~\cite{NOMA_JSTSP}. In contrast, we have used the \emph{exact} expression for $f_{g_{\min}}(x)$ to derive the closed-form expression for $R_{w, \mathrm{NOMA}}$ over $\alpha$-$\mu$ fading channels.
\subsection{High-SNR approximation for the ER of downlink NOMA}
Although the analytical expressions for the ER derived in the previous subsection are exact, they are computationally expensive. Therefore, we derive a high-SNR ($\rho \gg 1$) approximation for the ER of downlink NOMA in this subsection. For the case when $\rho \gg 1$, the ER of $U_s$ can be approximated as 
\begin{align*}
R_{s, \mathrm{NOMA}} \approx & \ - \dfrac{1}{\nu} \log_2 \left[ \mathbb E_{g_s}\left\{ (a_s \rho g_s)^{-\nu} \right\}\right] \\
= & \ \log_2(a_s \rho) - \dfrac{1}{\nu} \log_2 \left[ \int_0^\infty x^{-\nu} f_{g_s}(x) \mathrm dx\right]. 
\end{align*}
Substituting the expression for $f_{g_s}(x)$ from~\eqref{f_gi} into the preceding expression and solving the integral using~\cite[eqn.~(3.326-2),~p.~377]{Grad}, the preceding approximation can be written in closed-form as
	\begin{align}
		R_{s, \mathrm{NOMA}} \approx \log_2(a_s \rho) - \dfrac{1}{\nu} \log_2 \left[ \left( \dfrac{\mu^{1/\alpha}}{\Omega_s} \right)^{2\nu} \dfrac{\Gamma\left( \mu - \tfrac{2\nu}{\alpha}\right)}{\Gamma(\mu)}  \right], \label{Es_high}
	\end{align}
where the integral holds good only for the case when $\alpha \mu > 2 \nu$. On the other hand, it is straightforward to show that for $\rho \gg 1$, the ER of $U_w$ can be approximated as
	\begin{align}
		R_{w, \mathrm{NOMA}} \approx \log_2\left( 1 + \dfrac{a_w}{a_s}\right). \label{Ew_high}
	\end{align}
Note that the high-SNR approximation of $R_{w, \mathrm{NOMA}}$ does not depend on the QoS exponent $\theta$ or on the channel parameters ($\alpha$ and $\mu$). A high-SNR approximation for $R_{\mathrm{sum, NOMA}}$ can be obtained by adding~\eqref{Es_high} and~\eqref{Ew_high}. It can also be shown that the high-SNR approximation of the ER of $U_w$ is equal to the high-SNR approximation of the ergodic rate of $U_w$, which means that at high SNR, the effective rate of the weak user becomes (almost) equal to the ergodic rate, regardless of the delay QoS exponent.
\subsection{Low-SNR approximation for the ER of downlink NOMA}
Using a second-order Taylor expansion around $\rho \to 0$, an approximation for the ER of $U_i$ is given by 
	\begin{align}
		R_{i, \mathrm{NOMA}} = \rho  \dot R_{i, \mathrm{NOMA}} + 0.5 \rho^2  \ddot R_{i, \mathrm{NOMA}} + O(\rho^2), \label{Ei_LowSNR}
	\end{align}
where $\dot R_{i, \mathrm{NOMA}}$ and $\ddot R_{i, \mathrm{NOMA}}$ represent the first-order and second-order derivative of $R_{i, \mathrm{NOMA}}$ (shown in~\eqref{Ei_def}) with respect to $\rho$ and evaluated at $\rho \to 0$, respectively. 
\begin{theorem} \label{Theorem_Derivative}
The derivatives for $U_s$ and $U_w$ are given by 
\begin{equation}
	\begin{aligned}
		& \dot R_{s, \mathrm{NOMA}} = \log_2(e) a_s \mathbb E\{g_s\}, \\
		& \ddot R_{s, \mathrm{NOMA}} = \log_2(e) a_s^2 [\nu (\mathbb E\{g_s\})^2 - (\nu +1) \mathbb E\{g_s^2\}], \\ 
		& \dot R_{w, \mathrm{NOMA}} = \log_2(e) a_w \mathbb E\{g_{\min}\}, \\
		& \ddot R_{w, \mathrm{NOMA}} = \log_2(e) a_w \left[ \nu a_w (\mathbb E\{g_{\min}\})^2  \right. \\
		& \hspace{3cm}\left. - \{(\nu + 1) a_w + 2 a_s \} \mathbb E\{g_{\min}^2\}\right]. 
	\end{aligned} \label{Derivative}
\end{equation}
\end{theorem}
\begin{proof}
	See Appendix~\ref{Appendix_Theorem_Derivative}.
\end{proof}
It can be noted from~\eqref{Derivative} that $\dot R_{i, \mathrm{NOMA}}$ is independent of the delay QoS exponent $\theta$, and is in fact equal to the first derivative of the ergodic rate evaluated at $\rho \to 0$. It can also be noted from~\eqref{Derivative} that $\ddot R_{i, \mathrm{NOMA}}$ decreases as the value of $\theta$ increases\footnote{This occurs because $\nu \propto \theta$, and $\operatorname{Var}(g_s) > 0 \implies \mathbb E \{g_{s}^2\} > (\mathbb E\{g_s\})^2$. The same is also true for $g_{\min}$.} (i.e., for a system with more stringent delay constraints). It is also noteworthy that following the arguments in~\cite{Verdu}, it can be shown that the first derivative $\dot R_{i, \mathrm{NOMA}}$ is related to the minimum energy per bit required for reliable transmission of information, denoted by $\left( \mathcal E_{b, i}/\sigma^2 \right)_{\min}$, as follows:
\begin{equation}
 	\left( \dfrac{\mathcal E_{b, i}}{\sigma^2} \right)_{\min} = \dfrac{1}{\dot R_{i, \mathrm{NOMA}}}, 
\end{equation}
whereas the second derivative $\ddot R_{i, \mathrm{NOMA}}$ is related to the \emph{wideband slope}, denoted by $\mathcal S_{0, i}$, via
\begin{equation}
	\mathcal S_{0, i} = -2 \dfrac{(\dot R_{i, \mathrm{NOMA}})^2}{\ddot R_{i, \mathrm{NOMA}}} \ln 2. \label{S0}
\end{equation}
Therefore,~\eqref{Derivative}-\eqref{S0} imply that with an increase in the value of the delay QoS exponent, the minimum energy per bit required for reliable transmission does not change, but the  wideband slope vanishes. The relevant mean values in~\eqref{Derivative} can be obtained using~\eqref{f_gi},~\eqref{f_gmin} and~\cite[eqn.~(3.326-2),~p.~337]{Grad}, as follows: 
\begin{align*}
\mathbb E(g_s) = \dfrac{\Omega_s^2}{\mu^{2/\alpha} \Gamma(\mu)} \Gamma\left(\mu + \dfrac{2}{\alpha}\right), 
\end{align*}
\begin{align*}
	\mathbb E(g_s^2) = \dfrac{\Omega_s^4}{\mu^{4/\alpha} \Gamma(\mu)}\Gamma\left(\mu + \dfrac{4}{\alpha}\right), 
\end{align*}
\begin{multline*}
\mathbb E(g_{\min}) =  \dfrac{1}{\Omega_s^{\alpha \mu} \Gamma(\mu)} \sum_{m = 0}^{\mu - 1} \dfrac{\tilde \Omega^{\mu + m + (2/\alpha)}}{m! \mu^{2/\alpha}\Omega_w^{m\alpha}} \Gamma \left( \mu + m + \dfrac{2}{\alpha}\right) \\
	+ \dfrac{1}{\Omega_w^{\alpha \mu} \Gamma(\mu)} \sum_{n = 0}^{\mu - 1} \dfrac{\tilde \Omega^{\mu + n + (2/\alpha)}}{n! \mu^{2/\alpha}\Omega_s^{n\alpha}} \Gamma \left( \mu + n + \dfrac{2}{\alpha}\right), 
\end{multline*}
\begin{multline*}
\mathbb E(g_{\min}^2) = \dfrac{1}{\Omega_s^{\alpha \mu} \Gamma(\mu)} \sum_{m = 0}^{\mu - 1} \dfrac{\tilde \Omega^{\mu + m + (4/\alpha)}}{m! \mu^{4/\alpha}\Omega_w^{m\alpha}} \Gamma \left( \mu + m + \dfrac{4}{\alpha}\right) \\
	+ \dfrac{1}{\Omega_w^{\alpha \mu} \Gamma(\mu)} \sum_{n = 0}^{\mu - 1} \dfrac{\tilde \Omega^{\mu + n + (4/\alpha)}}{n! \mu^{4/\alpha}\Omega_s^{n\alpha}}\Gamma \left( \mu + n + \dfrac{4}{\alpha}\right) . 
\end{multline*}
Therefore, using~\eqref{Ei_LowSNR}, \eqref{Derivative} and the preceding expressions, one can find a closed-form low-SNR approximation for $R_{\mathrm{sum, NOMA}}$. 
\subsection{Upper bound on the ER of downlink NOMA}
Using the concavity of the logarithmic function and Jensen's inequality, it follows from~\eqref{Ei_def} that 
	\begin{align}
		R_{i, \mathrm{NOMA}} \leq & \ -\dfrac{1}{\nu} \mathbb E_{\gamma_i} \left\{ \log_2\left[ (1 + \gamma_i)^{-\nu} \right] \right\} \notag \\
		= & \ \mathbb E_{\gamma_i} \left\{ \log_2 (1 + \gamma_i)\right\} \triangleq C_{i, \mathrm{NOMA}}. \label{Ei_tilde}
	\end{align}
Note that $C_{i, \mathrm{NOMA}}$ does not depend on the delay exponent $\theta$ and represents the ergodic rate of $U_i$, and that equality holds for the case when $\theta \to 0$, i.e., a system without any delay QoS constraint. Using~\eqref{Ei_tilde}, the upper bound on the ER of $U_s$ is given by 
\begin{align*}
C_{s, \mathrm{NOMA}} = \int_0^\infty \log_2 (1 + a_s \rho x) f_{g_s}(x) \mathrm dx.
\end{align*}
Substituting the expression for $f_{g_s}(x)$ from~\eqref{f_gi} into the preceding equation and following the arguments in~Appendix~\ref{Appendix_Theorem_Mellin_s}, an analytical expression for $C_{s, \mathrm{NOMA}}$ is obtained as 
	\begin{align}
		C_{s, \mathrm{NOMA}} = \dfrac{\mu^{\mu} G_{2 \alpha, 2 + 2\alpha}^{2 + 2 \alpha, \alpha} \left[ \left. \tfrac{1}{(a_s \rho)^{\alpha}} \left(\tfrac{\mu}{2 \Omega_s^{\alpha}} \right)^2 \right\vert \begin{smallmatrix} \zeta, \ \chi \\[0.1em] \Delta(2, 0), \ \zeta, \ \zeta \end{smallmatrix}\right] }{\sqrt{2} (\ln 2) (2 \pi)^{\alpha-0.5} \Gamma(\mu) \Omega_s^{\alpha \mu} (a_s \rho)^{0.5 \alpha \mu} }, \label{Es_tilde}
	\end{align}
where $\zeta \triangleq \Delta(\alpha, -0.5 \alpha \mu)$.	Similarly, using~\eqref{Ei_tilde}, an upper bound on the ER of $U_w$ is given by 
\begin{multline*}
C_{w, \mathrm{NOMA}} = \int_0^\infty \log_2(1 + \rho x) f_{g_{\min}}(x) \mathrm dx \\
 - \int_0^\infty \log_2(1 + a_s \rho x) f_{g_{\min}}(x) \mathrm dx. 
\end{multline*}
Substituting the expression for $f_{g_{\min}}(x)$ and following the arguments in~Appendix~\ref{Appendix_Theorem_Mellin_s}, an analytical expression for $C_{w, \mathrm{NOMA}}$ is given by 
\begin{align}
	 & \!\!\!\! C_{w, \mathrm{NOMA}} = \dfrac{1}{\sqrt{2} (\ln 2) (2 \pi)^{\alpha - 0.5} \Gamma(\mu)}  \notag \\
	& \times \left[ \dfrac{1}{\Omega_s^{\alpha \mu}} \sum_{m = 0}^{\mu - 1} \dfrac{\mu^{\mu + m}}{m! \Omega_w^{m \alpha}} \left\{  \dfrac{\mathcal G(\rho, m)}{\rho^{0.5 \alpha(\mu + m)}} - \dfrac{\mathcal G(a_s \rho, m)}{(a_s \rho)^{0.5 \alpha(\mu + m)}} \right\} \right. \notag \\
	& \!\!\!\! \left. + \dfrac{1}{\Omega_w^{\alpha \mu}} \sum_{n = 0}^{\mu - 1} \dfrac{\mu^{\mu + n} }{n! \Omega_s^{n \alpha} } \left\{ \dfrac{\mathcal G(\rho, n)}{ \rho^{0.5 \alpha(\mu + n)} }  - \dfrac{\mathcal G(a_s \rho, n)}{(a_s \rho)^{0.5 \alpha(\mu + n)}} \right\} \right], \label{Ew_tilde}
\end{align}
where 
\begin{align*}
\mathcal G(x, y) = G_{2 \alpha, 2 + 2\alpha}^{2 + 2 \alpha, \alpha} \left[ \left. \dfrac{1}{x^{\alpha}} \left(\dfrac{\mu}{2 \tilde \Omega} \right)^2 \right\vert \begin{smallmatrix} \psi_y, \ \phi_y \\[0.6em] \Delta(2, 0), \ \psi_y, \ \psi_y\end{smallmatrix}\right],
\end{align*}
$\psi_y \triangleq \Delta(\alpha, -0.5 \alpha \{\mu + y\})$ and $\phi_y \triangleq \Delta(\alpha, 1-0.5 \alpha \{\mu + y\})$.
Therefore using~\eqref{Es_tilde} and~\eqref{Ew_tilde}, an analytical upper bound on the sum ER (representing the ergodic sum-rate of the downlink NOMA system) can be obtained as $ C_{\mathrm{sum, NOMA}} = C_{s, \mathrm{NOMA}} + C_{w, \mathrm{NOMA}}$.
\subsection{Effective rate of downlink OMA}
In order to quantify the performance difference between the downlink NOMA system and the corresponding OMA system, we analyze the ER of the downlink OMA system in this subsection. In the case of OMA, the source transmits $\sqrt{\mathcal E} x_i$ to $U_i$ in orthogonal time slots. Therefore, the instantaneous received SNR at $U_i$ to decode $x_i$ is given by $\hat \gamma_i = \rho g_i$. The ER for the user $U_i$ in the downlink OMA system is therefore given by 
\begin{align*}
R_{s, \mathrm{OMA}} = & \ -\dfrac{1}{\nu} \log_2 \left[ \mathbb E_{\gamma_s} \{(1 + \hat \gamma_s)^{-0.5 \nu}\} \right] \\
= & \ -\dfrac{1}{\nu} \log_2 \left[ \int_0^\infty (1 + \rho x)^{-0.5 \nu} f_{g_s}(x) \mathrm dx \right]. 
\end{align*}
Substituting the expression for $f_{g_s}(x)$ from~\eqref{f_gi} and following the arguments in~Appendix~\ref{Appendix_Theorem_Mellin_s} yields
\begin{align}
		R_{s, \mathrm{OMA}} = -\dfrac{1}{\nu} \log_2 \left[ \dfrac{\alpha^{0.5 \nu} \mu^{\mu} G_{\alpha, 2 + \alpha}^{2 + \alpha, \alpha} \left[ \left. \tfrac{\mu^2}{4 \rho^{\alpha} \Omega_s^{2\alpha}} \right\vert \begin{smallmatrix} \chi \\[0.1em] \Delta(2, 0), \ \hat \xi \end{smallmatrix} \right]}{ \sqrt{2} (2\pi)^{\alpha - 0.5} \Omega_s^{\alpha \mu} \Gamma(\mu) \Gamma(0.5 \nu) \rho^{0.5 \alpha \mu}}  \right], \label{Es_hat}
	\end{align}
where $\hat \xi \triangleq \Delta (\alpha, 0.5 \nu - 0.5 \alpha \mu)$. An analytical expression for the ER of $U_w$ in the downlink OMA system, denoted by $R_{w, \mathrm{OMA}} $, can be obtained by replacing $\Omega_s$ with $\Omega_w$ in~\eqref{Es_hat}. Finally, the sum ER can be obtained as $R_{\mathrm{sum, OMA}} = R_{s, \mathrm{OMA}} + R_{w, \mathrm{OMA}}$.
\section{Results and Discussion}
To analyze the performance of the system, we assume that $BT = 1$ and $\Omega_s = 1$. First we analyze the effect of $\Omega_w$ on the ER of NOMA, and also on the difference between the performance of the NOMA and OMA systems. Fig.~\ref{MeanValue} shows numerically evaluated results for the ER of the strong and weak user, as well as the sum ER of downlink NOMA, for $\Omega_s = 1$ and different values of $\Omega_w$. It can be noted from the figure that as the value of $\Omega_w$ increases, $R_{w, \mathrm{NOMA}}$ also increases in the low-to-mid SNR regime. This happens because an increase in the value of $\Omega_w$ results in a better (average) channel quality between the source $S$ and the weak user $U_w$. Interestingly, for high SNR values, $R_{w, \mathrm{NOMA}}$ becomes (almost) constant regardless of the value of $\Omega_w$, which is in line with~\eqref{Ew_high}. As for all the cases of $\Omega_w$ we have considered the value of $\Omega_s$ in the corresponding system to be fixed at $1$, the value of $R_{\mathrm{sum, NOMA}}$ is larger for larger values of $\Omega_w$ in the low-to-mid SNR regime and converges to the same value in the high-SNR regime. For the same system parameters, Fig.~\ref{NOMA-OMA_Diff_Omega} shows the numerically computed results on the difference between the sum ER of downlink NOMA and the corresponding downlink OMA system, i.e., $R_{\mathrm{sum, NOMA}}- R_{\mathrm{sum, OMA}}$. It can be noted form the figure that this difference is larger for smaller values of $\Omega_w$, i.e., the case when the channel quality of the $S$-$U_s$ link and $S$-$U_w$ links are significantly different. 
\begin{figure}[t]
\centering
  \includegraphics[width = 0.99\linewidth, height = 6.3cm]{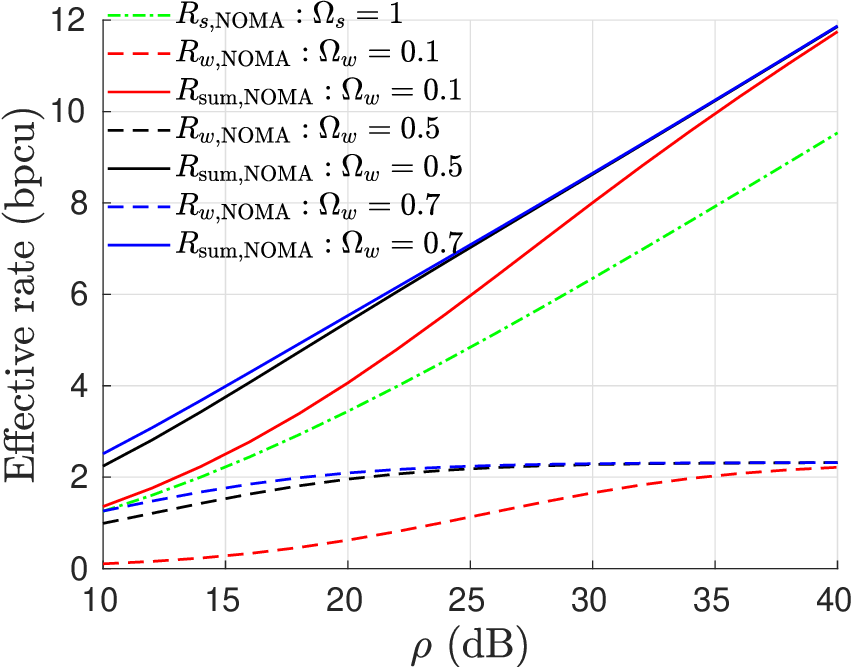}
  \caption{ER for the downlink NOMA system with $\alpha = \mu = 2$, $a_s = 0.2$ and $\theta = 1$.}
  \label{MeanValue}
\end{figure}
\begin{figure}[t]
  \centering
  \includegraphics[width = 0.99\linewidth, height = 6.3cm]{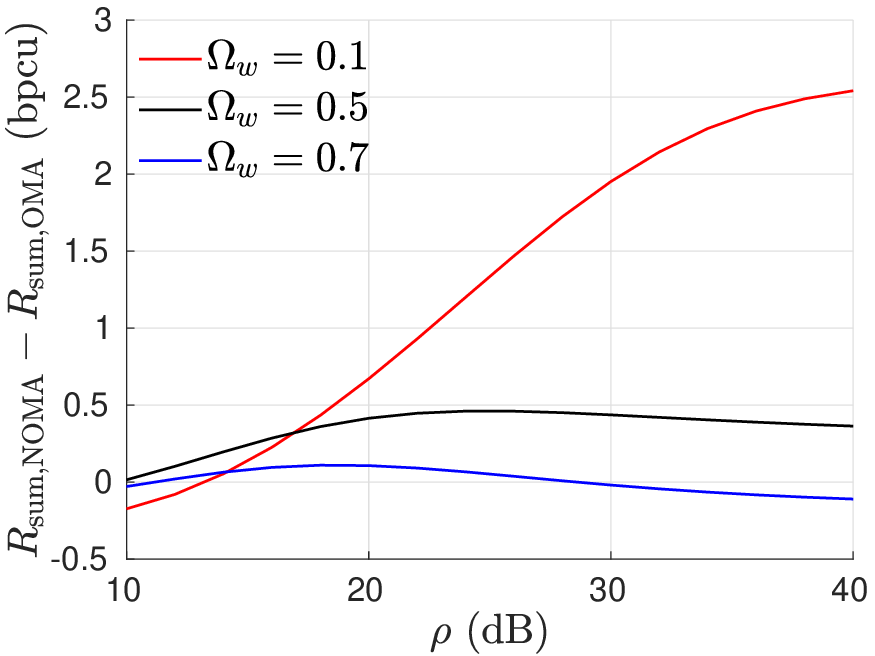}
  \caption{Difference between the sum ER of downlink NOMA and downlink OMA with $\alpha = \mu = 2$, $a_s = 0.2$, $\Omega_s = 1$ and $\theta = 1$.}
  \label{NOMA-OMA_Diff_Omega}
\end{figure}
This occurs because, for the case when $\Omega_s$ is fixed at $1$, although a larger value of $\Omega_w$ in the case of NOMA results in a higher sum ER in the low-to-mid SNR range, the corresponding increase in the ER of the weak user in the case of OMA is much larger, resulting in a smaller value of $R_{\mathrm{sum, NOMA}}- R_{\mathrm{sum, OMA}}$. Therefore, in the remaining figures in this paper, we will consider $\Omega_w = 0.1$. Assuming the target data rate for both $U_s$ and $U_w$ to be $r_{\mathrm{target}} = 2$ bpcu and following the arguments in~\cite{EL}, a legitimate range for $a_s$ is given by $0 < a_s < 1/2^r_{\mathrm{target}}$. Therefore, we use a one-dimensional numerical search over the discrete set $a_s \in \{0.01, 0.02, \ldots, 0.24\}$ in order to maximize the sum ER of downlink NOMA. For $\rho$ in the range from 10 dB to 40 dB, the optimal value of $a_s$ to maximize $R_{\mathrm{sum, NOMA}}$ is found to be 0.24, which is the maximum value considered for numerical evaluation. This occurs because most of the ER in $R_{\mathrm{sum, NOMA}}$ is obtained by $R_{s, \mathrm{NOMA}}$ (as the link between $S$ and $U_s$ is comparatively stronger). Therefore, throughout this section, we use $a_s = 0.24$.
\begin{figure}[t]
\centering
  \includegraphics[width = 0.99\linewidth, height = 6.3cm]{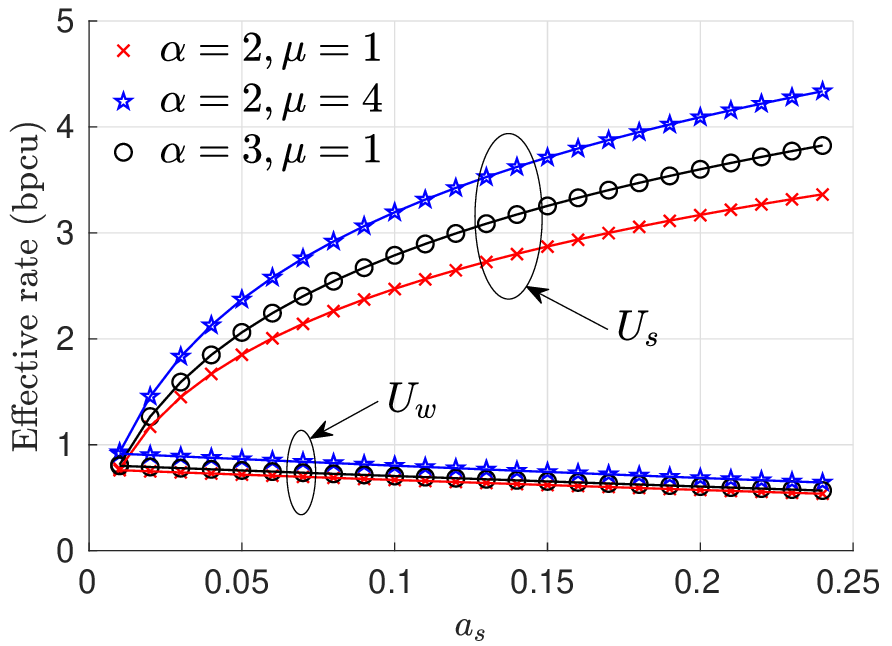}
  \caption{ER of the strong and weak users for $\theta = 0.5$ and $\rho = 20$ dB in downlink NOMA. Markers and (continuous) solid lines denote the numerically and analytically evaluated results, respectively.}
  \label{EC_Power}
\end{figure}
\begin{figure}[t]
  \centering
  \includegraphics[width = 0.99\linewidth, height = 6.3cm]{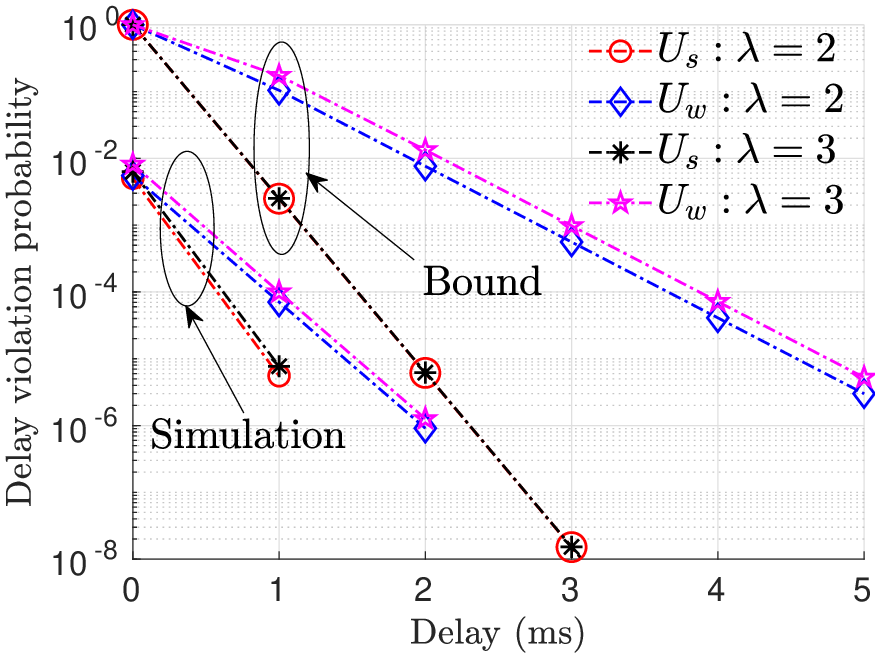}
  \caption{Delay violation probability for $\alpha = 2$, $\mu = 1$, $\rho = 10$ dB and $N = 168$ in downlink NOMA. Here the bounds are plotted using the derived analytical expressions.}
  \label{DelaySimulation}
\end{figure}
\begin{figure}[t]
\centering
  \includegraphics[width = 0.99\linewidth, height = 6.3cm]{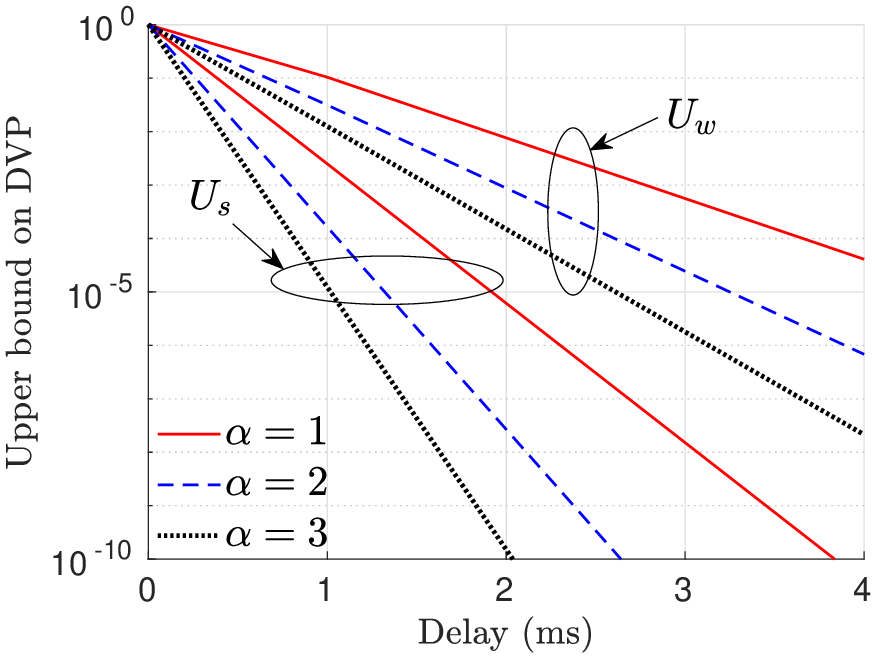}
  \caption{Analytically evaluated upper bound on the delay violation probability for $\mu = 1$, $\lambda = 2$, $\rho = 10$ dB and $N = 168$.}
  \label{DelayBound_alpha}
\end{figure}
\begin{figure}
  \centering
  \includegraphics[width = 0.99\linewidth, height = 6.3cm]{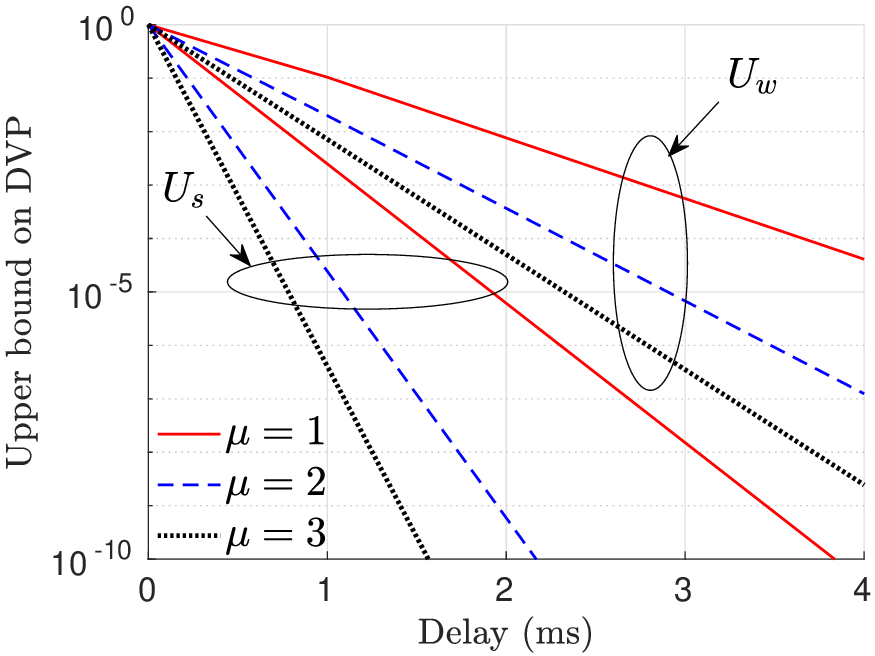}
  \caption{Analytically evaluated upper bound on the delay violation probability for $\alpha = 2$, $\lambda = 2$, $\rho = 10$ dB and $N = 168$.}
  \label{DelayBound_mu}
\end{figure}

\begin{figure}[t]
\centering
  \includegraphics[width = 0.99\linewidth, height = 6.3cm]{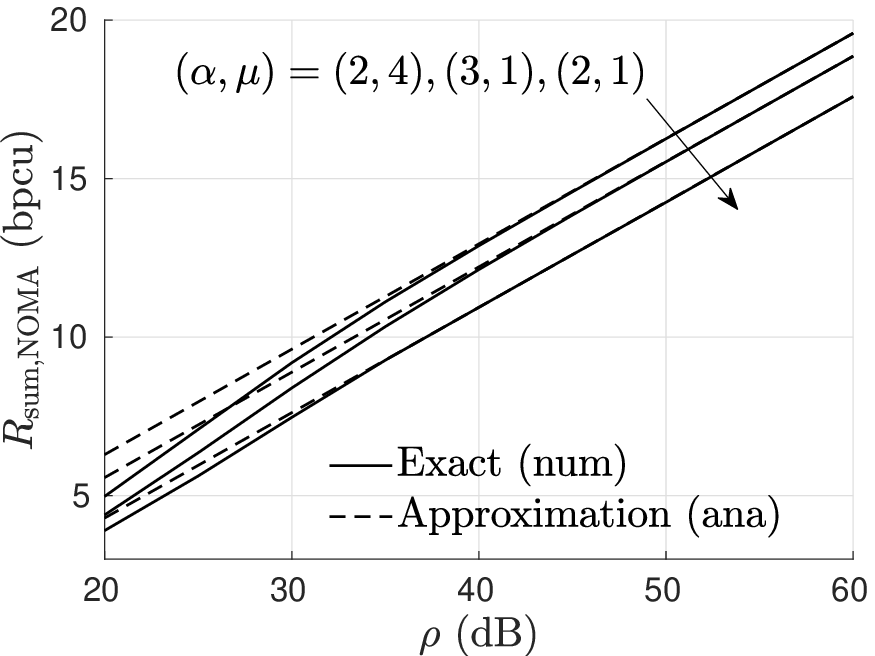}
  \caption{High-SNR approximation for the sum ER in downlink NOMA with delay QoS exponent $\theta = 0.5$.}
  \label{HighSNR}
\end{figure}
\begin{figure}[t]
  \centering
  \includegraphics[width = 0.99\linewidth, height = 6.3cm]{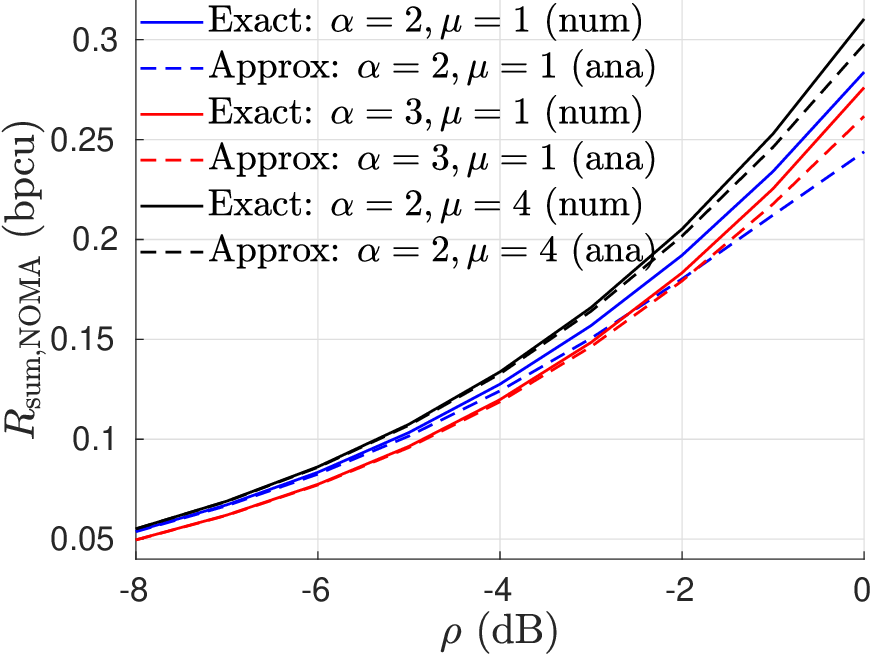}
  \caption{Low-SNR approximation for the sum ER in downlink NOMA with delay QoS exponent $\theta = 0.5$.}
  \label{LowSNR}
\end{figure}
\begin{figure}[t]
\centering
  \includegraphics[width = 0.99\linewidth, height = 6.3cm]{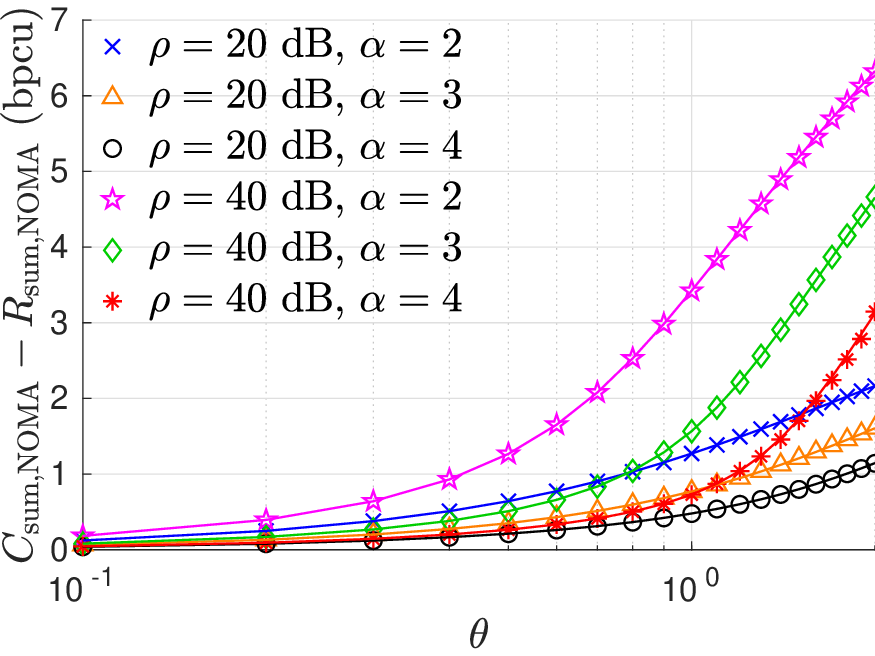}
  \caption{Difference between the ergodic sum-rate and sum ER in downlink NOMA for $\mu = 1$. Markers and (continuous) solid lines denote the numerically and analytically evaluated results, respectively.}
  \label{RateLoss_alpha}
\end{figure}
\begin{figure}
  \centering
  \includegraphics[width = 0.99\linewidth, height = 6.2cm]{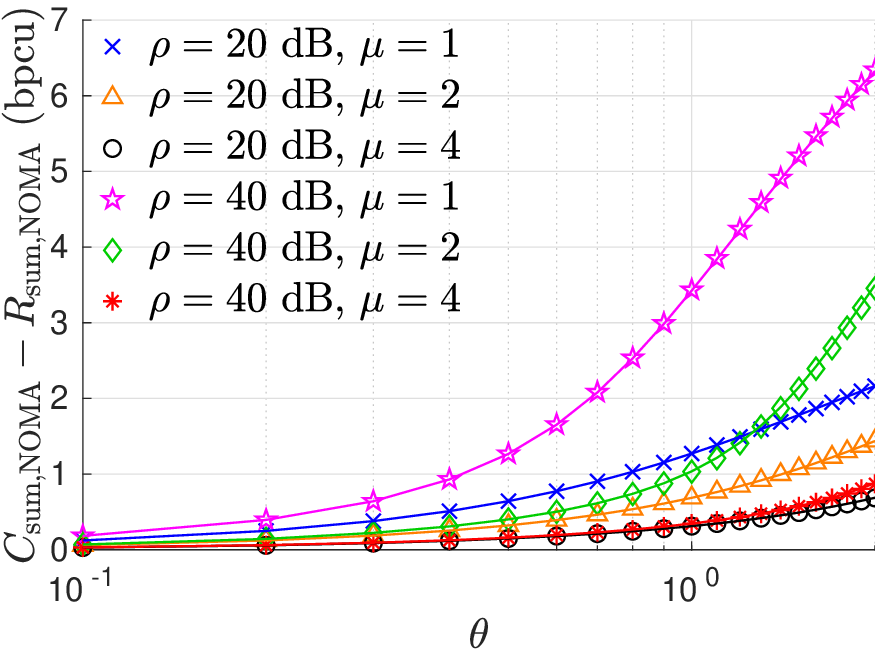}
  \caption{Difference between the ergodic sum-rate and sum ER in downlink NOMA for $\alpha = 2$. Markers and (continuous) solid lines denote the numerically and analytically evaluated results, respectively.}
  \label{RateLoss_mu}
\end{figure} 
Fig.~\ref{EC_Power} shows the variation of the ER of downlink NOMA for $U_s$ and $U_w$ for different $\alpha$ and $\mu$ w.r.t. the power-allocation coefficient. It can be noted from the figure that although the ER of $U_w$ decreases linearly with increasing value of $a_s$, the corresponding increase in the ER of $U_s$ is non-linear, resulting in a significant increase in $R_{\mathrm{sum, NOMA}}$. Similar results were shown in~\cite[Figs.~5~and~6]{SubOptimal} for the case of Rayleigh fading.

Fig.~\ref{DelaySimulation} shows the variation of the exact DVP and the corresponding upper bound w.r.t. the target delay $\vartheta$. It can be noted from the figure that the delay violation probability for $U_s$ is small as compared to that of $U_w$, because of the higher departure rate of the data intended for $U_s$. The figure also shows that with an increase in the data arrival rate (which results in a higher building rate of the queue), the delay violation probability increases for both $U_s$ and $U_w$. It is noteworthy that although the derived upper bounds on the probability of delay violation are not very tight (bounds of a similar nature were derived in~\cite{ACM_Paper},~\cite{FiniteBlockLength-DelayViolation},~\cite{LargeWirelessNetworks},~\cite{IoT-NOMA},~\cite{NOMA_JSTSP} and~\cite{UplinkNOMA-FiniteBlockLength}) the slope of the bounds are almost identical to that of the actual (simulated) curves. Therefore, these bounds can give good insights regarding the rate of decay of the delay violation probability curve. 

Figs.~\ref{DelayBound_alpha} and~\ref{DelayBound_mu} show the upper bound on the delay violation probability for different values of $\alpha$ (with $\mu$ fixed at $1$) and for different values of $\mu$ (with $\alpha$ fixed at 2), respectively. It can be noted from these figures that with an increase in the value of $\alpha$ or $\mu$, the probability of delay violation decreases at a higher rate. A large value of $\alpha$ (with fixed $\mu$) or $\mu$ (with fixed $\alpha$) represents a less severe fading\footnote{Note that for $\mu = 1$ the $\alpha$-$\mu$ distribution converges to the Weibull distribution with shape parameter $\alpha$, for which the amount of fading is given by $\texttt{AF} = [\Gamma(1 + \tfrac{4}{\alpha})/\{\Gamma(1 + \tfrac{2}{\alpha})\}^2] - 1$. Similarly, for $\alpha = 2$ the $\alpha$-$\mu$ fading model is equivalent to Nakagami-$m$ fading with shape parameter $m = \mu$, for which the amount of fading is given by $\texttt{AF} = 1/\mu$.}, resulting in a higher departure rate of the data at the queue at $S$. 

Figs.~\ref{HighSNR} and~\ref{LowSNR} show the approximations for the sum ER of the downlink NOMA system for $\rho \gg 1$ (high-SNR) and $\rho \to 0$ (low-SNR), respectively. The figures demonstrate a linear growth in the high-SNR regime and an exponential growth in the low-SNR regime for the sum ER. An excellent match between the exact (numerical) and the approximated (closed-form) results (in the relevant SNR regimes) confirms the correctness of the derived analytical approximations.

Figs.~\ref{RateLoss_alpha} and~\ref{RateLoss_mu} show the difference between the ergodic sum-rate and the sum ER (i.e., the rate loss $C_{\mathrm{sum, NOMA}} - R_{\mathrm{sum, NOMA}}$) for the downlink NOMA system, with varying delay QoS exponent $\theta$, for different values of $\alpha$ (with $\mu = 1$) and $\mu$ (with $\alpha = 2$), respectively. It can be noted from Fig.~\ref{RateLoss_alpha} that the rate loss decreases with increasing $\alpha$. Therefore, Fig.~\ref{RateLoss_alpha} indicates that a given delay QoS target will result in a more detrimental impact on the effective capacity in more severe fading conditions. Also, for a fixed value of $\alpha$ (with $\mu = 1$), the capacity loss becomes more pronounced as the SNR increases. Fig.~\ref{RateLoss_mu} shows the rate loss $C_{\mathrm{sum, NOMA}} - R_{\mathrm{sum, NOMA}}$ for different values of $\mu$ with $\alpha = 2$. Effects qualitatively similar to those observed in the previous case can be noticed here also; the loss in capacity increases with the severity of fading and also with the operational SNR.

\begin{figure}[t]
\centering
  \includegraphics[width = 0.99\linewidth, height = 6.3cm]{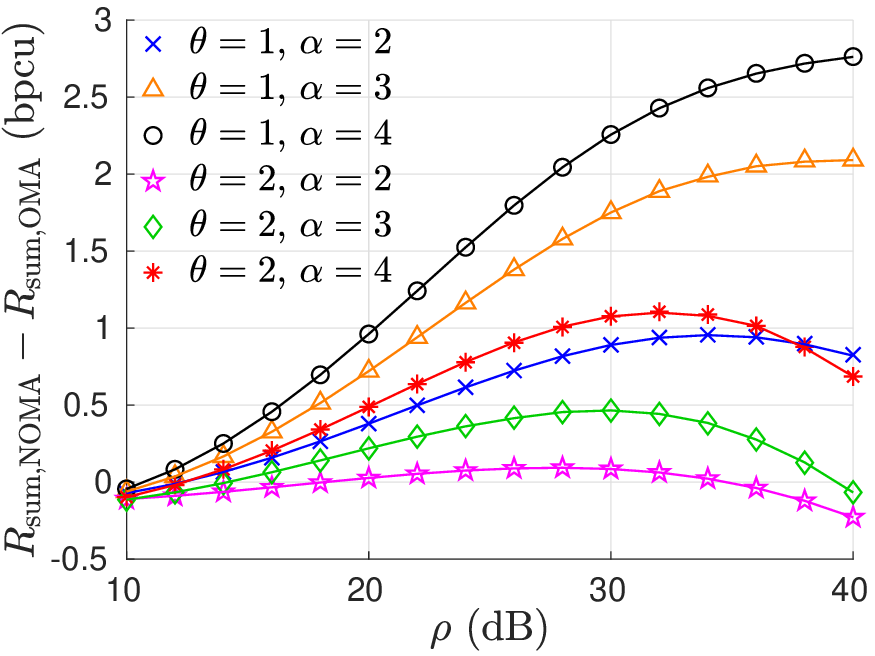}
  \caption{Difference between the sum ER in downlink NOMA and downlink OMA for $\mu = 1$. Markers and (continuous) solid lines denote the numerically and analytically evaluated results, respectively.}
  \label{EC_Diff_alpha}
\end{figure}
\begin{figure}[t]
  \centering
  \includegraphics[width = 0.99\linewidth, height = 6.2cm]{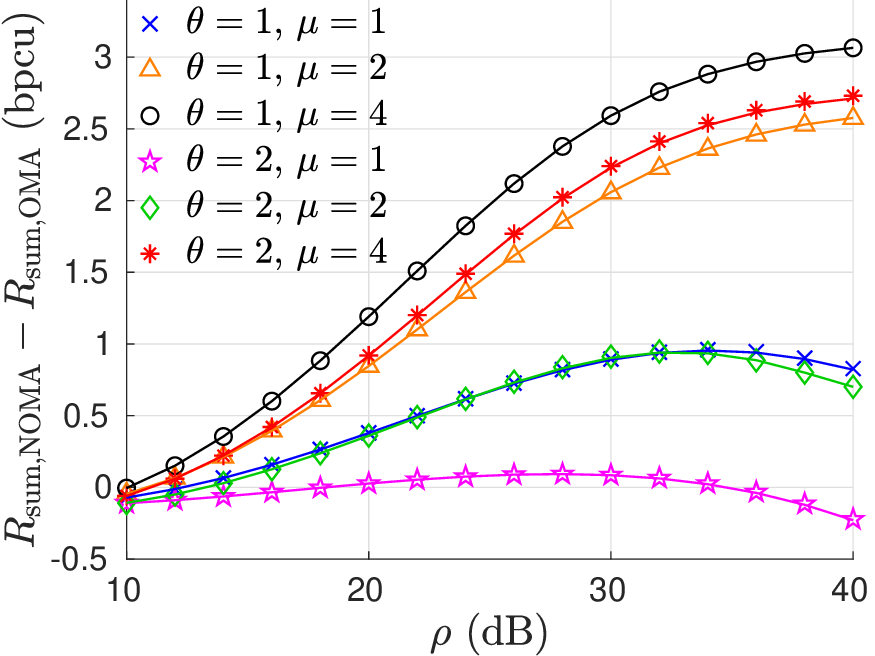}
  \caption{Difference between the sum ER in downlink NOMA and downlink OMA for $\alpha = 2$. Markers and (continuous) solid lines denote the numerically and analytically evaluated results, respectively.}
  \label{EC_Diff_mu}
\end{figure}
Figs.~\ref{EC_Diff_alpha} and~\ref{EC_Diff_mu} demonstrate the difference between the sum ER of downlink NOMA and the corresponding OMA system (i.e., $R_{\mathrm{sum, NOMA}} - R_{\mathrm{sum, OMA}}$) for different values of $\alpha$ (with $\mu = 1$) and $\mu$ (with $\alpha = 2$), respectively, with varying SNR $\rho$. It is clear from Fig.~\ref{EC_Diff_alpha} that for a fixed value of $\alpha$ and $\theta$, the performance gap between downlink NOMA and the corresponding OMA system increases (with NOMA outperforming OMA) in the low-to-mid-SNR regime and then this difference decreases in the high-SNR regime. Also, while NOMA outperforms OMA in most of the practical SNR range, the gain is smaller in more severe fading conditions (smaller $\alpha$ or $\mu$). It can also be seen from both figures that in the presence of a more stringent delay QoS constraint, the performance advantage of NOMA over OMA is considerably lessened. In particular, it can be seen that for the case of the stricter delay constraint ($\theta = 2$) and severe fading $(\alpha, \mu) = (2,1)$, the difference in achievable rate performance between NOMA and OMA is negligible.
\section{Conclusion}
In this paper, we analyzed the performance of a delay-constrained two-user downlink NOMA system over $\alpha$-$\mu$ fading. We derived an upper bound on the delay violation probability for each user, and an exact analytical expression for the sum effective rate of the downlink NOMA system over the generalized $\alpha$-$\mu$ channel, which as special cases yields the sum ER of the downlink NOMA system over a variety of channel models of practical interest, including Rayleigh, Nakagami-$m$, Weibull, Erlang, and Chi-squared fading. Analytical expressions for the approximations to the sum ER of downlink NOMA in high-SNR and low-SNR regimes were also derived, which are accurate in the corresponding regions. Simulation results indicate that the delay violation probability for the strong user is less compared to that of the weak user, and this probability of delay violation decreases with a decrease in the arrival data rate at the source as well as with the severity of fading. The results also indicate that for the case of NOMA, with increasing value of delay exponent the difference between ergodic sum-rate and sum ER increases very rapidly, but this difference decreases as the fading becomes less severe. Moreover, for a fixed channel parameters and delay exponent, this rate loss increases at high SNR. Finally, the difference between the sum ER of downlink NOMA and the corresponding OMA decreases for large values of delay exponent $\theta$ as well as under severe fading conditions.
\appendices
\section{Proof of Theorem~\ref{Theorem_Mellin_s}} \label{Appendix_Theorem_Mellin_s}
For the case of $U_s$, we have 
\begin{align*}
\mathcal M_{\varphi_s} (1 - \mathscr S) = & \  \mathbb E_{g_s}\{(1 + a_s \rho g_s)^{-\varpi}\} \\
= & \ \int_0^\infty (1 + a_s \rho x)^{-\varpi} f_{g_s} (x) \mathrm dx,
\end{align*}
where $\varpi \triangleq N \mathscr S/\ln 2$. Substituting the expression for $f_{g_s}(x)$ from~\eqref{f_gi} into the preceding equation, we obtain
\begin{multline*}
	\mathcal M_{\varphi_s} (1 - \mathscr S) = \dfrac{\alpha \mu^\mu}{2 \Omega_s^{\alpha \mu} \Gamma(\mu)} \int_0^\infty x^{0.5 \alpha \mu - 1} (1 + a_s \rho x)^{-\varpi} \\
	\times \exp \left(-\dfrac{\mu x^{0.5 \alpha}}{\Omega_s^{\alpha}}\right) \mathrm dx.
\end{multline*}
Using the relations $(1 + x)^{y} = \tfrac{1}{\Gamma(-y)}G_{1, 1}^{1, 1}\left[ x \left\vert \begin{smallmatrix}1 + y \\0 \end{smallmatrix}\right. \right]$ and $\exp(-x) = G_{0, 1}^{1, 0} \left[x \left\vert \begin{smallmatrix} - \\ 0 \end{smallmatrix}\right. \right]$ from~\cite[eqn.~(10)]{Reduce} and~\cite[eqn.~(11)]{Reduce}, respectively, it follows from the preceding equation that 
\begin{multline*}
	\mathcal M_{\varphi_s}(1 - \mathscr S) =  \dfrac{\alpha \mu^\mu}{2 \Omega_s^{\alpha \mu} \Gamma(\mu) \Gamma(\varpi)} \int_0^\infty x^{0.5 \alpha \mu - 1} \\
	\times G_{1, 1}^{1, 1} \left[ 1 + a_s \rho x \left \vert \begin{smallmatrix} 1 - \varpi \\[0.6em] 0\end{smallmatrix}\right.\right] G_{0, 1}^{1, 0} \left[\dfrac{\mu x^{0.5 \alpha}}{\Omega_s^{\alpha}} \left\vert \begin{smallmatrix} - \\[1em] 0 \end{smallmatrix} \right. \right] \mathrm dx.
\end{multline*}
Solving the integral given in the equation above using~\cite[eqn.~(21)]{Reduce}, an analytical expression for $\mathcal M_{\varphi_s}(1 - \mathscr S)$ can be given by~\eqref{Mellin_s}; this concludes the proof.
\section{Proof of Theorem~\ref{Theorem_Mellin_w}} \label{Appendix_Theorem_Mellin_w}
For the case of the weak user, we have 
\begin{align*}
\mathcal M_{\varphi_w}(1 - \mathscr S) = & \ \mathbb E_{g_{\min}}\left\{\left(1 + \dfrac{a_w \rho g_{\min}}{a_s \rho g_{\min} + 1}\right)^{-\varpi}\right\} \\
= & \ \int_0^\infty (1 + \rho x)^{-\varpi} (1 + a_s \rho x)^{\varpi} f_{g_{\min}}(x) \mathrm dx.
\end{align*}
Substituting the expression for $f_{g_{\min}}(x)$ from~\eqref{f_gmin} into the preceding equation, we obtain 
\begin{multline*}
\mathcal M_{\varphi_w}(1 - \mathscr S) = \dfrac{\alpha}{2\Omega_s^{\alpha \mu} \Gamma(\mu)} \sum_{m = 0}^{\mu - 1} \dfrac{\mu^{\mu + m}}{m! \Omega_w^{m \alpha}} \int_0^\infty (1 + \rho x)^{-\varpi} \\ 
\times (1 + a_s \rho x)^{\varpi} \exp \left( -\dfrac{\mu x^{0.5 \alpha}}{\tilde \Omega}\right) x^{0.5 \alpha(\mu + m) - 1} \mathrm dx \\
+ \dfrac{\alpha}{2\Omega_w^{\alpha \mu} \Gamma(\mu)} \sum_{n = 0}^{\mu - 1} \dfrac{\mu^{\mu + n}}{n! \Omega_s^{n \alpha}} \int_0^\infty (1 + \rho x)^{-\varpi} (1 + a_s \rho x)^{\varpi} \\
\times \exp \left( -\dfrac{\mu x^{0.5 \alpha}}{\tilde \Omega}\right) x^{0.5 \alpha(\mu + n) - 1} \mathrm dx.
\end{multline*}
Using~\cite[eqns.~(10) and~(11)]{Reduce}, it follows from the preceding equation that 
\begin{multline*}
\mathcal M_{\varphi_w}(1 - \mathscr S) = \dfrac{\alpha}{2\Omega_s^{\alpha \mu} \Gamma(\mu) \Gamma(\varpi) \Gamma(-\varpi)} \sum_{m = 0}^{\mu - 1} \dfrac{\mu^{\mu + m}}{m! \Omega_w^{m \alpha}} \\ 
\times \int_0^\infty \!\!\!\! x^{0.5 \alpha(\mu + m) - 1}  G_{1, 1}^{1, 1} \left[ 1 \!+\! \rho x \left\vert \begin{smallmatrix} 1 - \varpi \\[0.6em] 0\end{smallmatrix}\right. \!\! \right] G_{1, 1}^{1, 1} \left[ 1 \!+\! a_s \rho x \left\vert \begin{smallmatrix} 1 + \varpi \\[0.6em] 0\end{smallmatrix}\right.\!\! \right]  \\
 \times  G_{0, 1}^{1, 0} \left[ \dfrac{\mu x^{0.5 \alpha}}{\tilde \Omega} \left\vert \begin{smallmatrix} - \\[1em] 0 \end{smallmatrix} \right. \!\! \right] \mathrm dx + \dfrac{\alpha}{2\Omega_w^{\alpha \mu} \Gamma(\mu) \Gamma(\varpi) \Gamma(-\varpi)} \sum_{n = 0}^{\mu - 1} \dfrac{\mu^{\mu + n}}{n! \Omega_s^{n \alpha}} \\
\times \int_0^\infty \!\!\!\! x^{0.5 \alpha(\mu + n) - 1} G_{1, 1}^{1, 1} \left[ 1 \!+\! \rho x \left\vert \begin{smallmatrix} 1 - \varpi \\[0.6em] 0\end{smallmatrix}\right. \!\! \right] G_{1, 1}^{1, 1} \left[ 1 \!+\! a_s \rho x \left\vert \begin{smallmatrix} 1 + \varpi \\[0.6em] 0\end{smallmatrix}\right. \!\! \right] \\
\times G_{0, 1}^{1, 0} \left[ \dfrac{\mu x^{0.5 \alpha}}{\tilde \Omega} \left\vert \begin{smallmatrix} - \\[1em] 0 \end{smallmatrix} \right. \right] \mathrm dx.
\end{multline*}
Using the relation between Meijer's G-function and Fox's H-function given in~\cite[eqn.~(6.2.8)]{Springer}, the preceding equation yields
\begin{multline*}
\mathcal M_{\varphi_w}(1 - \mathscr S) =  \dfrac{\alpha}{2\Omega_s^{\alpha \mu} \Gamma(\mu) \Gamma(\varpi) \Gamma(-\varpi)} \sum_{m = 0}^{\mu - 1} \dfrac{\mu^{\mu + m}}{m! \Omega_w^{m \alpha}} \\
\times \int_0^\infty  x^{0.5 \alpha(\mu + m) - 1} H_{1, 1}^{1, 1} \left[ 1 + \rho x \left\vert \begin{smallmatrix} (1 - \varpi, 1) \\[0.6em] (0, 1)\end{smallmatrix}\right.\right] \\
\times  H_{1, 1}^{1, 1} \left[ 1 + a_s \rho x \left\vert \begin{smallmatrix} (1 + \varpi, 1) \\[0.6em] (0, 1)\end{smallmatrix}\right.\right] H_{0, 1}^{1, 0} \left[ \dfrac{\mu x^{0.5 \alpha}}{\tilde \Omega} \left\vert \begin{smallmatrix} - \\[1em] (0, 1) \end{smallmatrix} \right. \right] \mathrm dx \\
+ \dfrac{\alpha}{2\Omega_w^{\alpha \mu} \Gamma(\mu) \Gamma(\varpi) \Gamma(-\varpi)} \sum_{n = 0}^{\mu - 1} \dfrac{\mu^{\mu + n}}{n! \Omega_s^{n \alpha}} \int_0^\infty x^{0.5 \alpha(\mu + n) - 1} \\
\times H_{1, 1}^{1, 1} \left[ 1 + \rho x \left\vert \begin{smallmatrix} (1 - \varpi, 1) \\[0.6em] (0, 1)\end{smallmatrix}\right.\right] H_{1, 1}^{1, 1} \left[ 1 + a_s \rho x \left\vert \begin{smallmatrix} (1 + \varpi, 1) \\[0.6em] (0, 1)\end{smallmatrix}\right.\right] \\
\times H_{0, 1}^{1, 0} \left[ \dfrac{\mu x^{0.5 \alpha}}{\tilde \Omega} \left\vert \begin{smallmatrix} - \\[1em] (0, 1) \end{smallmatrix} \right. \right] \mathrm dx.
\end{multline*}
Using the substitution $\tilde x = x^{0.5 \alpha}$ and solving the resulting integrals using~\cite[eqn.~(2.3)]{Mittal}, the analytical expression for $\mathcal M_{\varphi_w} (1 - \mathscr S)$ reduces to~\eqref{Mellin_w}; this completes the proof.
\begin{figure*}
\vskip-0.15in
\section{Proof of Theorem~\ref{Theorem_Derivative}} \label{Appendix_Theorem_Derivative}
For the case of the strong user, from~\eqref{Ei_def} it follows that
\begin{align*}
		 \dot R_{s, \mathrm{NOMA}} \triangleq & \ \left. \dfrac{\mathrm d R_{s, \mathrm{NOMA}}}{\mathrm d \rho} \right\vert_{\rho \to 0} = \left. \dfrac{\mathrm d}{\mathrm d \rho} \left( -\dfrac{1}{\nu} \log_2 \left[ \mathbb E_{g_s} \left\{ (1 + a_s \rho g_s)^{\nu}\right\}\right]\right) \right\vert_{\rho \to 0} \\
		= & \ \left. -\dfrac{\log_2(e)}{\nu} \times \dfrac{\mathbb E_{g_s} \left\{ -\nu (1 + a_s \rho g_s)^{-\nu - 1} a_s g_s \right\}}{\mathbb E_{g_s} \left\{ (1 + a_s \rho g_s)^{-\nu}\right\}}\right\vert_{\rho \to 0} = \log_2(e) a_s \mathbb E\{g_s\},
\end{align*}
and
\begin{align*}
		& \ddot R_{s, \mathrm{NOMA}} \triangleq \left. \dfrac{\mathrm d^2 R_{s, \mathrm{NOMA}}}{\mathrm d \rho^2} \right\vert_{\rho \to 0} \\
		= & \! \left. -\dfrac{\log_2 (e)}{\nu} \times \dfrac{\mathbb E_{g_s} \left\{(1 + a_s \rho g_s)^{-\nu}\right\} \mathbb E_{g_s}\! \left\{ \nu (\nu + 1) (1 + a_s \rho g_s)^{-\nu - 2} a_s^2 g_s^2 \right\} \!-\! \left( \mathbb E_{g_s} \left\{ -\nu (1 + a_s \rho g_s)^{-\nu - 1} a_s g_s\right\} \right)^2}{\left( \mathbb E_{g_s} \left\{ (1 + a_s \rho g_s)^{-\nu}\right\}\right)^2} \right\vert_{\rho \to 0} \\
		= & \ -\dfrac{\log_2(e)}{\nu} \left[ \nu (\nu + 1) a_s^2 \mathbb E \left\{g_s^2\right\} - \nu^2 a_s^2 \left( \mathbb E\{g_s\} \right)^2\right] = \log_2(e) a_s^2 \left[ \nu \left( \mathbb E \{g_s\}\right)^2 - (\nu + 1) \mathbb E \left\{ g_s^2\right\}\right].
\end{align*}
Similarly, for the weak user, it follows from~\eqref{Ei_def} that 
\begin{align*}
		& \dot R_{w, \mathrm{NOMA}} \triangleq \left. \dfrac{\mathrm d R_{w, \mathrm{NOMA}}}{\mathrm d \rho} \right\vert_{\rho \to 0} = \left. \dfrac{\mathrm d}{\mathrm d \rho} \left( -\dfrac{1}{\nu} \log_2 \left[ \mathbb E_{g_{\min}} \left\{ \left( \dfrac{1 + \rho g_{\min}}{1 + a_s \rho g_{\min}} \right)^{-\nu} \right\} \right]\right) \right\vert_{\rho \to 0} \\
		= & \ -\dfrac{\log_2(e)}{\nu} \left[\mathbb E_{g_{\min}} \left\{ \left(\dfrac{1 + \rho g_{\min}}{1 + a_s \rho g_{\min}} \right)^{-\nu} \right\} \right]^{-1} \\
		& \hspace{4cm}\times \left. \mathbb E_{g_{\min}} \left\{ -\nu \left(\dfrac{1 + \rho g_{\min}}{1 + a_s \rho g_{\min}} \right)^{-\nu - 1} \dfrac{(1 + a_s \rho g_{\min}) g_{\min} - a_s g_{\min} (1 + \rho g_{\min})}{(1 + a_s \rho g_{\min})^2}\right\} \right\vert_{\rho \to 0} \\
		= & \ -\dfrac{\log_2(e)}{\nu} \mathbb E_{g_{\min}} \left\{ -\nu g_{\min} (1 - a_s)\right\} = \log_2(e) a_w \mathbb E \left\{ g_{\min} \right\}, 
\end{align*}
and
\begin{align*}
	& \ddot R_{w, \mathrm{NOMA}} \triangleq \left. \dfrac{\mathrm d^2 R_{w, \mathrm{NOMA}}}{\mathrm d \rho^2} \right\vert_{\rho \to 0} = \dfrac{-\log_2(e)}{\nu} \left[ \mathbb E_{g_{\min}} \left\{ \left( \dfrac{1 + \rho g_{\min}}{1 + a_s \rho g_{\min}}\right)^{-\nu}\right\}\right]^{-2} \left[ \mathbb E_{g_{\min}} \left\{ \left( \dfrac{1 + \rho g_{\min}}{1 + a_s \rho g_{\min}}\right)^{-\nu}\right\}\right. \\
	& \times \mathbb E_{g_{\min}} \left\{ \nu (\nu + 1) \left( \dfrac{1 + \rho g_{\min}}{1 + a_s \rho g_{\min}}\right)^{-\nu - 2} \left(\dfrac{(1 + a_s \rho g_{\min}) g_{\min} - a_s g_{\min} (1 + \rho g_{\min})}{(1 + a_s \rho g_{\min})^2}\right)^2 \!- \!\nu \left(\dfrac{1 + \rho g_{\min}}{1 + a_s \rho g_{\min}} \right)^{-\nu - 1}\right. \\
	& \times \left. \dfrac{(1 + a_s \rho g_{\min})^2 (a_s g_{\min}^2 - a_s g_{\min}^2) - 2 (1 + a_s \rho g_{\min}) a_s g_{\min} [(1 + a_s \rho g_{\min}) g_{\min} - a_s g_{\min} (1 + \rho g_{\min})]}{(1 + a_s \rho g_{\min})^4} \right\} \\
	& - \left. \left. \left( \mathbb E_{g_{\min}} \left\{ -\nu \left( \dfrac{1 + \rho g_{\min}}{1 + a_s \rho g_{\min}}\right)^{-\nu - 1} \dfrac{(1 + a_s \rho g_{\min}) g_{\min} - a_s g_{\min}(1 + \rho g_{\min})}{(1 + a_s \rho g_{\min})^2}\right\} \right)^2 \right] \right\vert_{\rho \to 0}\\
	= & \dfrac{-\log_2(e)}{\nu} \left[ \mathbb E_{g_{\min}} \left\{ \nu (\nu + 1) (g_{\min} - a_s g_{\min})^2 + 2 \nu a_s g_{\min} (g_{\min} - a_s g_{\min})\right\} - \left( \mathbb E_{g_{\min}}\left\{-\nu (g_{\min} - a_s g_{\min}) \right\}\right)^2\right] \\
	= & \log_2 (e) a_w \left[ \nu a_w \left( \mathbb E \left\{ g_{\min}\right\} \right)^2 - \left\{ (\nu + 1) a_w + 2 a_s\right\} \mathbb E \left\{ g_{\min}^2\right\} \right].
\end{align*}
\hrulefill
\vskip-0.15in
\end{figure*}
\Urlmuskip=0mu plus 1mu\relax
\bibliographystyle{IEEEtran}
\bibliography{AlphaMuEC}

\begin{thebibliography}{10}
\providecommand{\url}[1]{#1}
\csname url@samestyle\endcsname
\providecommand{\newblock}{\relax}
\providecommand{\bibinfo}[2]{#2}
\providecommand{\BIBentrySTDinterwordspacing}{\spaceskip=0pt\relax}
\providecommand{\BIBentryALTinterwordstretchfactor}{4}
\providecommand{\BIBentryALTinterwordspacing}{\spaceskip=\fontdimen2\font plus
\BIBentryALTinterwordstretchfactor\fontdimen3\font minus
  \fontdimen4\font\relax}
\providecommand{\BIBforeignlanguage}[2]{{%
\expandafter\ifx\csname l@#1\endcsname\relax
\typeout{** WARNING: IEEEtran.bst: No hyphenation pattern has been}%
\typeout{** loaded for the language `#1'. Using the pattern for}%
\typeout{** the default language instead.}%
\else
\language=\csname l@#1\endcsname
\fi
#2}}
\providecommand{\BIBdecl}{\relax}
\BIBdecl

\bibitem{Negi}
{Dapeng Wu} and R.~{Negi}, ``Effective capacity: A wireless link model for
  support of quality of service,'' \emph{IEEE Trans. Wireless Commun.}, vol.~2,
  no.~4, pp. 630--643, July 2003.

\bibitem{CR-Nakagami}
L.~{Musavian} and S.~{Aissa}, ``Effective capacity of delay-constrained
  cognitive radio in {N}akagami fading channels,'' \emph{IEEE Trans. Wireless
  Commun.}, vol.~9, no.~3, pp. 1054--1062, 2010.

\bibitem{Unified-MISO}
M.~{You}, H.~{Sun}, J.~{Jiang}, and J.~{Zhang}, ``Unified framework for the
  effective rate analysis of wireless communication systems over {MISO} fading
  channels,'' \emph{IEEE Trans. Commun.}, vol.~65, no.~4, pp. 1775--1785, 2017.

\bibitem{Weibull}
------, ``Effective rate analysis in {W}eibull fading channels,'' \emph{IEEE
  Wireless Commun. Lett.}, vol.~5, no.~4, pp. 340--343, 2016.

\bibitem{EtaMu}
J.~{Zhang}, M.~{Matthaiou}, Z.~{Tan}, and H.~{Wang}, ``Effective rate analysis
  of {MISO} $\eta$-$\mu$ fading channels,'' in \emph{IEEE Int. Conf. Commun.
  (ICC)}, 2013, pp. 5840--5844.

\bibitem{AlphaMu}
J.~{Zhang}, L.~{Dai}, Z.~{Wang}, D.~W.~K. {Ng}, and W.~H. {Gerstacker},
  ``Effective rate analysis of {MISO} systems over $\alpha$-$\mu$ fading
  channels,'' in \emph{IEEE Global Commun. Conf. (GLOBECOM)}, 2015, pp. 1--6.

\bibitem{KappaMuShadowed}
J.~{Zhang}, L.~{Dai}, W.~H. {Gerstacker}, and Z.~{Wang}, ``Effective capacity
  of communication systems over $\kappa$-$\mu$ shadowed fading channels,''
  \emph{Electron. Lett.}, vol.~51, no.~19, pp. 1540--1542, 2015.

\bibitem{Composite}
H.~{Al-Hmood} and H.~S. {Al-Raweshidy}, ``Unified approaches based effective
  capacity analysis over composite $\alpha$-$\eta$-$\mu$/{G}amma fading
  channels,'' \emph{Electron. Lett.}, vol.~54, no.~13, pp. 852--853, 2018.

\bibitem{Snedecor}
S.~{Chen}, J.~{Zhang}, G.~K. {Karagiannidis}, and B.~{Ai}, ``Effective rate of
  {MISO} systems over {Fisher}-{Snedecor} $\mathcal{F}$ fading channels,''
  \emph{IEEE Commun. Lett.}, vol.~22, no.~12, pp. 2619--2622, December 2018.

\bibitem{VTC}
K.~D. {Kanellopoulou}, K.~P. {Peppas}, and P.~T. {Mathiopoulos}, ``Effective
  capacity analysis of equal gain diversity combiners over generalized fading
  channels,'' in \emph{IEEE Veh. Techol. Conf. (VTC Spring)}, June 2018, pp.
  1--5.

\bibitem{LpNorm}
------, ``Effective capacity of ${L}_p$-norm diversity receivers over
  generalized fading channels under adaptive transmission schemes,'' \emph{IEEE
  Trans. Commun.}, vol.~68, no.~2, pp. 1240--1253, 2020.

\bibitem{ACM_Paper}
S.~Schiessl, J.~Gross, and H.~Al-Zubaidy, ``Delay analysis for wireless fading
  channels with finite blocklength channel coding,'' in \emph{ACM International
  Conference on Modeling, Analysis and Simulation of Wireless and Mobile
  Systems}, 2015, pp. 13--22.

\bibitem{FiniteBlockLength-DelayViolation}
S.~{Schiessl}, H.~{Al-Zubaidy}, M.~{Skoglund}, and J.~{Gross}, ``Delay
  performance of wireless communications with imperfect {CSI} and finite-length
  coding,'' \emph{IEEE Trans. Commun.}, vol.~66, no.~12, pp. 6527--6541, 2018.

\bibitem{MISO-FiniteBlockLength-DelayViolation}
S.~{Schiessl}, J.~{Gross}, M.~{Skoglund}, and G.~{Caire}, ``Delay performance
  of the multiuser {MISO} downlink under imperfect {CSI} and finite-length
  coding,'' \emph{IEEE J. Sel. Areas Commun.}, vol.~37, no.~4, pp. 765--779,
  2019.

\bibitem{AoI}
N.~{Pappas} and M.~{Kountouris}, ``Delay violation probability and age of
  information interplay in the two-user multiple access channel,'' in
  \emph{IEEE Int. Wksp Signal Proc. Adv. Wireless Commun. (SPAWC)}, 2019, pp.
  1--5.

\bibitem{LargeWirelessNetworks}
M.~{Kountouris}, N.~{Pappas}, and A.~{Avranas}, ``{QoS} provisioning in large
  wireless networks,'' in \emph{Int. Symp. Modeling and Optimization in Mobile,
  Ad Hoc, and Wireless Networks (WiOpt)}, May 2018, pp. 1--6.

\bibitem{SubOptimal}
J.~{Choi}, ``Effective capacity of {NOMA} and a suboptimal power control policy
  with delay {QoS},'' \emph{IEEE Trans. Commun.}, vol.~65, no.~4, pp.
  1849--1858, April 2017.

\bibitem{QoS_Provisioning_Ding}
G.~{Liu}, Z.~{Ma}, X.~{Chen}, Z.~{Ding}, F.~R. {Yu}, and P.~{Fan},
  ``Cross-layer power allocation in nonorthogonal multiple access systems for
  statistical {QoS} provisioning,'' \emph{IEEE Trans. Veh. Technol.}, vol.~66,
  no.~12, pp. 11\,388--11\,393, 2017.

\bibitem{Half-FullDuplex-ER}
X.~{Chen}, G.~{Liu}, and Z.~{Ma}, ``Statistical {QoS} provisioning for
  half/full-duplex cooperative non-orthogonal multiple access,'' in \emph{IEEE
  Veh. Techol. Conf. (VTC-Fall)}, 2017, pp. 1--5.

\bibitem{Mussavian}
W.~{Yu}, L.~{Musavian}, and Q.~{Ni}, ``Link-layer capacity of {NOMA} under
  statistical delay {QoS} guarantees,'' \emph{IEEE Trans. Commun.}, vol.~66,
  no.~10, pp. 4907--4922, October 2018.

\bibitem{IoT-NOMA}
C.~{Xiao}, J.~{Zeng}, B.~{Liu}, X.~{Su}, and J.~{Wang}, ``Cross-layer power
  control for uplink {NOMA} in {IoT} applications with statistical delay
  constraints,'' in \emph{IEEE Global Commun. Conf. (GLOBECOM)}, 2018, pp.
  1--7.

\bibitem{NOMA_JSTSP}
C.~{Xiao}, J.~{Zeng}, W.~{Ni}, R.~P. {Liu}, X.~{Su}, and J.~{Wang}, ``Delay
  guarantee and effective capacity of downlink {NOMA} fading channels,''
  \emph{IEEE J. Sel. Topics Signal Proc.}, vol.~13, no.~3, pp. 508--523, June
  2019.

\bibitem{Poor}
X.~{Zhang}, J.~{Wang}, and H.~V. {Poor}, ``{NOMA}-based statistical {QoS}
  provisioning for wireless ad-hoc networks with finite blocklength,'' in
  \emph{IEEE Global Commun. Conf. (GLOBECOM)}, 2019, pp. 1--6.

\bibitem{NOMA_Secrecy}
W.~{Yu}, A.~{Chorti}, L.~{Musavian}, H.~V. {Poor}, and Q.~{Ni}, ``Effective
  secrecy rate for a downlink {NOMA} network,'' \emph{IEEE Trans. Wireless
  Commun.}, vol.~18, no.~12, pp. 5673--5690, December 2019.

\bibitem{Uplink-Musavian}
M.~Bello, W.~Yu, A.~Chorti, and L.~Musavian, ``Performance analysis of {NOMA}
  uplink networks under statistical {QoS} delay constraints,'' \emph{arXiv
  preprint arXiv:2003.04758}, 2020.

\bibitem{FiniteBlockLength-Musavian}
M.~Amjad, L.~Musavian, and S.~A{\"\i}ssa, ``Effective capacity of {NOMA} with
  finite blocklength for low-latency communications,'' \emph{arXiv preprint
  arXiv:2002.07098}, 2020.

\bibitem{UplinkNOMA-FiniteBlockLength}
S.~{Schiessl}, M.~{Skoglund}, and J.~{Gross}, ``{NOMA} in the uplink: {D}elay
  analysis with imperfect {CSI} and finite-length coding,'' \emph{IEEE Trans.
  Wireless Commun.}, to appear.

\bibitem{Yacoub}
M.~D. {Yacoub}, ``The $\alpha$-$\mu$ distribution: {A} physical fading model
  for the {S}tacy distribution,'' \emph{IEEE Trans. Veh. Technol.}, vol.~56,
  no.~1, pp. 27--34, 2007.

\bibitem{ICC}
V.~{Kumar}, B.~{Cardiff}, and M.~F. {Flanagan}, ``Performance analysis of
  {NOMA}-based cooperative relaying in $\alpha$-$\mu$ fading channels,'' in
  \emph{2019 IEEE Int. Conf. Commun. (ICC)}, May 2019, pp. 1--7.

\bibitem{Grad}
A.~Jeffrey and D.~Zwillinger, \emph{Table of Integrals, Series, and Products},
  7th~ed.\hskip 1em plus 0.5em minus 0.4em\relax Elsevier Science, 2007.

\bibitem{Secrecy}
H.~Lei, I.~S. Ansari, G.~Pan, B.~Alomair, and M.~Alouini, ``Secrecy capacity
  analysis over $\alpha - \mu $ fading channels,'' \emph{IEEE Commun. Lett.},
  vol.~21, no.~6, pp. 1445--1448, June 2017.

\bibitem{MatlabImplement}
K.~P. Peppas, ``A new formula for the average bit error probability of dual-hop
  amplify-and-forward relaying systems over generalized shadowed fading
  channels,'' \emph{IEEE Wireless Commun. Lett.}, vol.~1, no.~2, pp. 85--88,
  April 2012.

\bibitem{Verdu}
S.~{Verd\'u}, ``Spectral efficiency in the wideband regime,'' \emph{IEEE Trans.
  Inf. Theory}, vol.~48, no.~6, pp. 1319--1343, 2002.

\bibitem{EL}
V.~{Kumar}, B.~{Cardiff}, and M.~F. {Flanagan}, ``Performance analysis of
  {NOMA} with generalised selection combining receivers,'' \emph{Electron.
  Lett.}, vol.~55, no.~25, pp. 1364--1367, December 2019.

\bibitem{Reduce}
V.~S. Adamchik and O.~I. Marichev, ``The algorithm for calculating integrals of
  hypergeometric type functions and its realization in {REDUCE} system,'' in
  \emph{Proceedings of the International Symposium on Symbolic and Algebraic
  Computation}.\hskip 1em plus 0.5em minus 0.4em\relax New York, NY, USA: ACM,
  1990, pp. 212--224.

\bibitem{Springer}
M.~D. Springer, \emph{The Algebra of Random Variables}.\hskip 1em plus 0.5em
  minus 0.4em\relax John Wiley \& Sons, 1979.

\bibitem{Mittal}
P.~K. Mittal and K.~C. Gupta, ``An integral involving generalized function of
  two variables,'' \emph{Proceedings of the Indian Academy of Sciences -
  Section A}, vol.~75, no.~3, pp. 117--123, Mar 1972.

\end{thebibliography}
\begin{IEEEbiography}[{\includegraphics[width=1in,height=1.25in,clip,keepaspectratio]{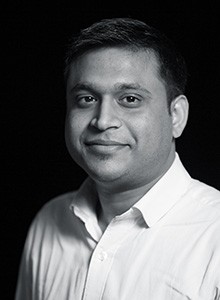}}]%
{Vaibhav Kumar} (S'17) received the B.E. degree in electronics and telecommunications engineering from CSVTU, Bhilai, India in 2012, the M.Tech. degree  in electronics and communications engineering from The LNM Institute of Information Technology, Jaipur, India in 2015, and the Ph.D. degree in electronic engineering from University College Dublin, Ireland in 2020. 

From 2012 to 2013, he worked as a lecturer in the Dept. of Biomedical Engineering, National Institute of Technology, Raipur, India. From 2015 to 2016, he worked as a project associate in the project entitled ``Mobile broadband service support over cognitive radio networks,'' funded by Media Lab Asia, Govt. of India. He was a visiting research student at Indian Institute of Technology~--~Delhi, India, from Jan 2019 to May 2019, under the Erasmus+~ICM research program. Since April 2020, he is working as a postdoctoral research fellow at University College Dublin, Ireland. His research interests include wireless communication theory, physical-layer issues in cognitive radio networks, physical-layer network coding and non-orthogonal multiple access techniques.
\end{IEEEbiography}
\begin{IEEEbiography}[{\includegraphics[width=1in,height=1.25in,clip,keepaspectratio]{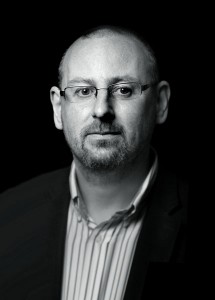}}]%
{Barry Cardiff} (M'06--SM'19) received the B.Eng., M.Eng.Sc., and Ph.D. degrees in electronic engineering from University College Dublin, Ireland, in 1992, 1995 and 2011, respectively. 

He was a Senior Design Engineer and Systems Architect for Nokia
from 1993 to 2001, moving to Silicon \& Software Systems (S3 group) thereafter as a Systems Architect in their R\&D division focused on wireless communications and digitally assisted circuit design. Since 2013 he is an Assistant Professor at University College Dublin, Ireland. His research interests are in digitally assisted circuit design and signal processing for wireless and optical communication systems. He holds several US patents related to wireless communication. 
\end{IEEEbiography}
\begin{IEEEbiography}[{\includegraphics[width=1in,height=1.25in,clip,keepaspectratio]{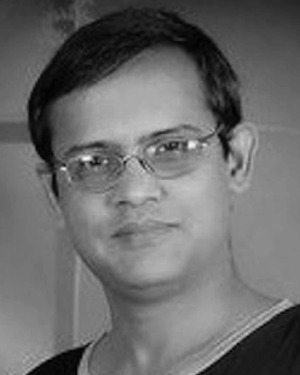}}]%
{Shankar Prakriya} (SM'02) received the B.E. degree (Hons) in electronics and communication engineering from the Regional Engineering College, Bharathidasan University, Tiruchirappalli, in 1987, and the M.A.Sc. (Engg.) and Ph.D. degrees from the Department of Electrical and Computer Engineering, University of Toronto, Toronto, ON, Canada, in 1993 and 1997, respectively. He was with the
Indian Space Research Organization for about three years. 

In 1997, he joined Indian Institute of Technology Delhi (IITD), where he is currently a Professor with the Department of Electrical Engineering. He was the Jai Gupta Research Chair Professor for five years until 2017. He holds three U.S. and some Indian patents. He has served in the technical program committee of the prominent IEEE international conferences. His research interests include in wireless energy harvesting, interference management, non-orthogonal multiple access, and full-duplex communication.
\end{IEEEbiography}
\begin{IEEEbiography}[{\includegraphics[width=1in,height=1.25in,clip,keepaspectratio]{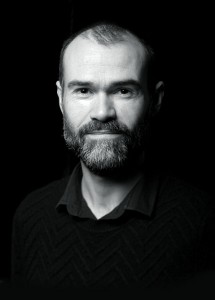}}]%
{Mark F. Flanagan} (M'03--SM'10) received the B.E. and Ph.D. degrees in electronic engineering from University College Dublin, Ireland, in 1998 and 2005, respectively.

From 1998 to 1999, he was a Project Engineer with Parthus Technologies Ltd. From 2006 to 2008, he was a Post-Doctoral Research Fellow with the University of Zurich, Switzerland, the University of Bologna, Italy, and The University of Edinburgh, U.K. In 2008, he was appointed as an SFI Stokes Lecturer in electronic engineering with University College Dublin, where he is currently an Associate Professor. In 2014, he was a Visiting Senior Scientist with the Institute of Communications and Navigation, German Aerospace Center, under a DLR-DAAD Fellowship. His research interests include information theory, wireless communications, and signal processing.

Dr. Flanagan is a Senior Member of the IEEE Communications Society and the IEEE Signal Processing Society. He has served on the technical program committees of several IEEE international conferences, and he was TPC co-chair for the Communication Theory Symposium at ICC 2020. He is currently serving as an Executive Editor for \textsc{IEEE Communications Letters}, and as an Associate Editor for \textsc{Frontiers in Communications and Networks}.
\end{IEEEbiography}
\end{document}